\documentclass[11pt,reqno,a4paper,final]{amsart}
\usepackage[british]{babel}
\usepackage{amssymb,amstext,amsthm,eucal,bbm,amscd,bbold}
\usepackage{mathtools}
\usepackage{tensor}
\usepackage[margin=1in]{geometry}
\usepackage{array}
\usepackage[svgnames,table]{xcolor}
\usepackage[T1]{fontenc}
\usepackage[utf8]{inputenc}
\usepackage[small]{eulervm}
\usepackage[scr=boondox, cal=esstix]{mathalpha}
\usepackage{tgpagella}
\usepackage{mathrsfs}
\usepackage{verbatim}
\usepackage{pdfsync}
\usepackage[final]{microtype}
\usepackage[nosort]{cite}
\usepackage{adjustbox}
\usepackage{tabularx}
\usepackage{multirow}
\usepackage{booktabs} 
\usepackage{makecell}
\usepackage{caption}
\usepackage[toc]{appendix}
\usepackage{chngcntr}
\usepackage{apptools}
\usepackage{romanbar,braket}
\usepackage{enumitem}
\usepackage{slashed}
\usepackage[boxsize=2pt,aligntableaux=bottom]{ytableau}

\usepackage{tikz,pgfplots}
\pgfplotsset{compat=1.18}
\usepgfplotslibrary{fillbetween}
\usetikzlibrary{calc,arrows,arrows.meta,cd,shapes.geometric,positioning,through,decorations.markings,decorations.text}
\tikzset{>=latex}
\usepackage[unicode,pagebackref=true]{hyperref}
\hypersetup{%
  pdftitle   = {Non-relativistic quantum strings from gauged WZW models},
  pdfkeywords = {Non-relativistic, string theory, CFT, BRST, Virasoro},
  pdfauthor  = {Jos\'{e} Figueroa-O'Farrill, Girish S Vishwa},
  pdfcreator = {\LaTeX\ with package \flqq hyperref\frqq},
  linkcolor=MediumBlue,
  citecolor=DarkGreen,
  urlcolor=OrangeRed,
  anchorcolor=OrangeRed,
  colorlinks=true}
\renewcommand*{\backref}[1]{}
\renewcommand*{\backrefalt}[4]{%
  \ifcase #1%
  \or [Page~#2.]%
  \else [Pages~#2.]%
  \fi%
}
\usepackage{cleveref}
\theoremstyle{plain}
\newtheorem{lemma}{Lemma}
\newtheorem{proposition}[lemma]{Proposition}
\newtheorem{theorem}[lemma]{Theorem}
\newtheorem{corollary}[lemma]{Corollary}

\theoremstyle{definition}
\newtheorem{definition}[lemma]{Definition}

\newcommand{\normalord}[1]{\pmb{(}#1\pmb{)}}
\newcommand{\g}{\mathfrak{g}}
\newcommand{\bms}{\mathfrak{bms}}
\newcommand{\ghat}{\widehat{\mathfrak{g}}}

\newcommand{\h}{\mathfrak{h}}

\renewcommand{\a}{\mathfrak{a}}
\renewcommand{\b}{\mathfrak{b}}
\renewcommand{\d}{\partial}

\newcommand{\nw}{\mathfrak{nw}}
\newcommand{\so}{\mathfrak{so}}

\newcommand{\fsl}{\mathfrak{sl}}
\newcommand{\s}{\mathfrak{s}}

\newcommand{\B}{\mathscr{B}}

\newcommand{\x}{\boldsymbol{x}}
\newcommand{\y}{\boldsymbol{y}}

\renewcommand{\k}{\boldsymbol{k}}

\newcommand{\Gr}{\operatorname{Gr}}

\newcommand{\ad}{\operatorname{ad}}

\newcommand{\reg}{\operatorname{reg}}
\newcommand{\Tgh}{T^{\mathrm{gh}}}

\newcommand{\dNG}{d_\mathrm{NG}}

\newcommand{\jVir}{j_\mathrm{Vir}}
\newcommand{\dVir}{d_\mathrm{Vir}}
\newcommand{\Tmatter}{T^\mathrm{mat}}

\newcommand{\Tbarmatter}{\bar{T}^\mathrm{mat}}

\newcommand{\Tsug}{T^\mathrm{sug}}
\newcommand{\barpar}{\bar{\partial}}
\newcommand{\betabar}{\bar{\beta}}
\newcommand{\gammabar}{\bar{\gamma}}
\newcommand{\betil}{\widetilde{\beta}}
\newcommand{\gamtil}{\widetilde{\gamma}}
\newcommand{\betilbar}{\bar{\widetilde{\beta}}}
\newcommand{\gamtilbar}{\bar{\widetilde{\gamma}}}
\newcommand{\bbar}{\bar{b}}
\newcommand{\cbar}{\bar{c}}
\newcommand{\Tbargh}{\bar{T}^{\mathrm{gh}}}
\newcommand{\fdeg}{\operatorname{fdeg}}

\newcommand{\Ad}{\operatorname{Ad}}

\newcommand{\Tr}{\operatorname{Tr}}

\newcommand{\RR}{\mathbb{R}}

\newcommand{\NN}{\mathbb{N}}

\newcommand{\ZZ}{\mathbb{Z}}

\newcommand{\CC}{\mathbb{C}}

\newcommand{\Vir}{\operatorname{Vir}}
\newcommand{\superbracket}[1]{[\kern-.15em[{#1}]\kern-.15em]}
\newcommand{\eqdef}{=\vcentcolon}
\newcommand{\defeq}{\vcentcolon=}

\newcommand{\AdS}{\operatorname{AdS}}

\newcommand{\End}{\operatorname{End}}

\newcommand{\im}{\operatorname{im}}

\newcommand{\Xbar}{\overline{X}}

\newcommand{\dotr}{\mbox{$\boldsymbol{\cdot}$}}

\newcommand{\bigslant}[2]{\left.\raisebox{.25em}{$#1$}\scalebox{1.5}{/}\raisebox{-.25em}{$#2$}\right.}
\definecolor{dkgr}{rgb}{0,0.6,0}
\definecolor{gris}{rgb}{0.5,0.5,0.5}

\allowdisplaybreaks[0]
\numberwithin{equation}{section}
\AtAppendix{\counterwithin{lemma}{section}}

\begin{document}
\title{Non-relativistic quantum strings from gauged WZW models}

\author[Figueroa-O'Farrill]{José M Figueroa-O'Farrill}
\author[Vishwa]{Girish S Vishwa}
\address{Maxwell Institute and School of Mathematics, The University
  of Edinburgh, James Clerk Maxwell Building, Peter Guthrie Tait Road,
  Edinburgh EH9 3FD, Scotland, United Kingdom}
\email[JMF]{\href{mailto:j.m.figueroa@ed.ac.uk}{j.m.figueroa@ed.ac.uk}, ORCID: \href{https://orcid.org/0000-0002-9308-9360}{0000-0002-9308-9360}}
\email[GSV]{\href{mailto:}{G.S.Vishwa@sms.ed.ac.uk}, ORCID: \href{https://orcid.org/0000-0001-5867-7207}{0000-0001-5867-7207}}

\begin{abstract}
  We construct non-relativistic quantum strings from gauged
  Wess--Zumino--Witten (WZW) models. We depart from the fact that Lie
  groups with a bi-invariant galilean structure can be seen as the
  quotient by a null central subgroup of a generalised Nappi--Witten
  group. We implement the quotient by a chiral null gauging.  We use a
  particular free field realisation of the Nappi--Witten current
  algebra to compute the Virasoro BRST cohomology of the gauged WZW
  model, resulting in a closed string theory reminiscent of the
  bosonic Gomis--Ooguri string, but whose spectrum differs slightly
  between the holomorphic and antiholomorphic sectors.
\end{abstract}
\maketitle

\tableofcontents

\section{Introduction}

The quarter century since the pioneering paper of Gomis and Ooguri
\cite{Gomis:2000bd} (see also \cite{Danielsson:2000gi}) on so-called
non-relativistic strings has witnessed a surge of activity on string
theories where either the target space or the worldsheet (or both) are
non-lorentzian; that is, where the relevant geometric structure is not
given by a metric (regardless the signature).  Progress in this area
has been recently reviewed in \cite{Oling:2022fft}.  As with much else
in the non-lorentzian approach to fundamental physics, one can obtain
models of non-lorentzian string theories via so-called
non-relativistic \cite{Gomis:2005pg, Batlle:2016iel, Hartong:2021ekg,
  Hartong:2022dsx, Hartong:2024ydv} or so-called ultra-local
\cite{Cardona:2016ytk, Bagchi:2023cfp, Blair:2023noj, Gomis:2023eav,
  Harksen:2024bnh, Bagchi:2024rje} limits of the Nambu--Goto or
Polyakov strings, but there are also constructions of non-relativistic
strings via null reduction of the Nambu--Goto \cite{Harmark:2017rpg}
or Polyakov \cite{Harmark:2018cdl, Bidussi:2021ujm} actions, and also
intrinsic constructions on stringy Newton--Cartan geometries
\cite{Bergshoeff:2018yvt, Bergshoeff:2019pij, Yan:2019xsf}.  A paper
comparing some of these approaches is \cite{Harmark:2019upf}.  There
are other constructions of non-relativistic strings in the literature:
for example in \cite{Brugues:2004an} they are constructed via
Wess--Zumino terms on the Galilei group and
in \cite{Fontanella:2022fjd,Fontanella:2022pbm} they are constructed
via sigma models on homogeneous spaces of stringy extended Galilei
and Newton--Hooke groups.

The aim of this paper is to exhibit yet a different construction of
non-relativistic quantum strings, inspired on the construction of
galilean spacetimes by null reduction of a lorentzian manifold
\cite{PhysRevD.31.1841,Julia:1994bs,Bekaert:2015xua}.  Because strings
propagating on Lie groups are well understood as conformal field
theories of Wess--Zumino--Witten (WZW) type, we will take the
lorentzian manifold to be a Lie group with a bi-invariant lorentzian
metric having a null subgroup that we can gauge.  Gauged WZW models
are two-dimensional conformal field theories amenable to a
quantum-mechanical treatment based on vertex operator algebras and
this allows for a wholly quantum mechanical treatment of the resulting
non-relativistic string theory.

Lie groups admitting a bi-invariant lorentzian metric were classified
(up to coverings) by Medina in \cite{MR814072}.  This is roughly
equivalent to the classification of Lie algebras admitting an
ad-invariant lorentzian inner product, which we will refer to as
\emph{lorentzian Lie algebras} from now on.  Every lorentzian Lie
algebra is either indecomposable or an orthogonal direct sum of an
indecomposable lorentzian Lie algebra with a Lie algebra admitting an
ad-invariant positive-definite inner product.  We call these latter
Lie algebras \emph{euclidean} from now on.  The list of indecomposable
lorentzian Lie algebras is small: it is either $\fsl(2,\RR)$, whose
Killing form has lorentzian signature, or else it is a double
extension \cite{MR826103} (to be described below) of an
even-dimensional abelian euclidean Lie algebra by a one-dimensional
Lie algebra spanned by a (maximal rank) skew-symmetric derivation.
The paradigmatic example is the Nappi--Witten Lie algebra
\cite{Nappi:1993ie} (cf. Definition \ref{def:gen nw alg}) which is the double
extension of a two-dimensional abelian Lie algebra with a euclidean
inner product by the one-dimensional Lie algebra of skew-symmetric
derivations, and which will play a crucial role in this paper.

This paper is organised as follows.  In
Section~\ref{sec:preliminaries} we review the characterisation of Lie
algebras with ad-invariant galilean, carrollian and bargmannian
structures, following \cite{Figueroa-OFarrill:2022pus}, and review the
definition of the WZW model.  It is there that we define the notion of
a generalised Nappi--Witten Lie algebra.  The bulk of the paper deals
with WZW models on generalised Nappi--Witten Lie groups.  In
Section~\ref{sec:gauged-wzw-models} we review the gauging of WZW
models and in particular the obstructions to doing so and apply this
to the null gauging of a WZW model based on a generalised
Nappi--Witten group.  In Section~\ref{sec:Implement gauge constraint
  via BRST} we implement the null gauging of the Nappi--Witten WZW
model homologically using the BRST formalism in conformal field
theory.  For ease of computation we make two choices: firstly, we
choose to work with a different real form of the complexified
Nappi--Witten Lie algebra and, secondly, we choose a particular
free-field realisation of this Nappi--Witten current algebra.  These
choices allow us to compute the null gauging BRST cohomology to arrive
at a conformal field theory, which when extended with matter fields in
order to reach the critical central charge, we interpret as a
non-relativistic string theory.  We compute the Virasoro BRST
cohomology, which displays a slight heterosis between the holomorphic
and antiholomorphic sectors, but which nevertheless is reminiscent of
the closed Gomis--Ooguri bosonic string.  Section~\ref{sec:discussion}
contains a discussion and some comments about our results, as well as
pointing to possible future work.  The paper has a number of
appendices: Appendix~\ref{app:spectral
  sequences} on spectral sequences and Appendix~\ref{app:KO quartet
  mechanism} on the Kugo--Ojima mechanism are expository, while
Appendix~\ref{sec:cohom-comp} on the
cohomology proofs whose end results were only stated in
Section~\ref{sec:Implement gauge constraint via BRST} and
Appendix~\ref{sec:null-reds} on the null reduced geometries contain the results of calculations.

\section{Preliminaries}
\label{sec:preliminaries}

We start by recalling the basic notions of galilean, carrollian and
bargmannian Lie algebras (following \cite{Figueroa-OFarrill:2022pus})
and how they are related, as well as the basic definition of a
WZW model.

\subsection{Galilean, carrollian and bargmannian structures on Lie groups}
\label{sec:galcarbarg}

There is a rich geometric story to these structures (see, for example,
the recent review \cite{Bergshoeff:2022eog}).  We are interested in
Lie groups admitting bi-invariant such geometric structures.  Modulo
the usual ambiguity inherent in the Lie correspondence, we prefer to
work with their Lie algebras.  Lie algebras with ad-invariant galilean
carrollian and bargmannian structures are intimately linked and, as
shown in \cite{Figueroa-OFarrill:2022pus}, there is a bijection (a
sort of duality) between galilean and carrollian Lie algebras,
mediated by bargmannian Lie algebras, which we now recall.

We start by recalling the notion of a metric Lie algebra, first
systematically introduced by Medina and Revoy \cite{MR826103}.

\begin{definition}
  A \emph{metric Lie algebra} $\g_0$ is a real Lie algebra admitting a
  symmetric bilinear inner product $\langle-,-\rangle_0$ that is
  invariant under the adjoint action of $\g_0$ on itself. That is, for
  all $X, Y, Z \in \g_0$, 
  \begin{equation*}
    \langle \ad_X Y, Z \rangle_0 + \langle Y, \ad_X Z \rangle_0 = 0;
  \end{equation*}
  equivalently,
  \begin{equation*}
    \left< [Y,X],Z\right> = \left<Y, [X,Z]\right>.
  \end{equation*}
\end{definition}

Cartan's semisimplicity criterion says that every semisimple Lie
algebra is metric relative to its Killing form
$\left<X,Y\right> := \Tr \ad_X \ad_Y$.  At the other extreme, given
any real Lie algebra $\h$ we may give the vector space
$\h \oplus \h^*$ the structure of metric Lie algebra as follows: $\h$
acts on $\h^*$ via the coadjoint representation, whereas $\h^*$
becomes an abelian ideal.  The inner product is simply the dual
pairing between $\h$ and $\h^*$, which has split signature in this
case.  This latter example is the simplest example of a double
extension, which we now recall.

Let $(\g_0,\left<-,-\right>_0)$ be a metric Lie algebra and let
$\so(\g_0)$ denote the Lie algebra of skew-symmetric endomorphisms of
$\g_0$.  Let $\h$ be a Lie algebra acting on $\g_0$ via skew-symmetric
derivations, with $\partial \colon \h \to \so(\g_0)$ denoting the
corresponding Lie algebra homomorphism.  This defines a skew-symmetric
bilinear map $\omega \colon \g_0 \times \g_0 \to \h^*$ by transposition; namely,
\begin{equation}
  \omega(X,Y)(h) = \left<\partial(h)X, Y\right>_0.
\end{equation}
The vector space $\g := \g_0 \oplus \h \oplus \h^*$ now becomes a Lie
algebra under the bracket
\begin{equation}
  [(X,h,\alpha), (Y,k,\beta)] = \left([X,Y]_{\g_0} + \partial(h) Y - \partial(k)X,
   [h,k]_{\h}, \ad_h^* \beta - \ad_h^* \alpha + \omega(X,Y)\right)
\end{equation}
for all $X,Y \in \g_0$, $h,k \in \h$ and $\alpha,\beta \in \h^*$.
Moreover, $\g$ is metric relative to the inner product
\begin{equation}
  \left<(X,h,\alpha), (Y,k,\beta)\right> = \left<X,Y\right>_0 +
  \alpha(k) + \beta(h) + \left<h,k\right>_{\h},
\end{equation}
where $\left<-,-\right>_{\h}$ is any invariant symmetric bilinear
form (not necessarily an inner product) on $\h$.  The metric Lie
algebra $(\g,\left<-,-\right>$ is said to be a \emph{double
  extension} of the metric Lie algebra $(\g_0,\left<-,-\right>_0$ by
$\h$.  If $\left<-,-\right>_0$ has signature $(p,q)$ and $\h$ has
dimension $r$, then $\left<-,-\right>$ has signature $(p+r,q+r)$.

Given two metric Lie algebras, we can construct a third metric Lie
algebra by taking their orthogonal direct sum.  A metric Lie algebra
which is not isomorphic to an orthogonal direct sum of metric Lie
algebras is said to be \emph{indecomposable}.  The fundamental theorem
of metric Lie algebras \cite[Théorème~II]{MR826103} says that the
class of metric Lie algebras is the smallest class which contains the
simple and one-dimensional Lie algebras and which is closed under the
operations of orthogonal direct sum and double extension.

It follows from this theorem and the above observation about the
signature of the inner product of a double extension, that a euclidean
Lie algebra $\g$ is the orthogonal direct sum
\begin{equation}
  \label{eq:metric-LA-pos-def}
  \g = \s_1 \oplus \dots \oplus \s_n \oplus \a,
\end{equation}
where the $\s_i$ are compact simple Lie algebras, with inner product
given by a negative multiple of the Killing form, and $\a$ is a
euclidean abelian Lie algebra.

A precursor theorem of Medina's \cite[Thm.~4.1]{MR814072} states that
any lorentzian Lie algebra is either indecomposable or else the
orthogonal direct sum of an indecomposable lorentzian Lie algebra with
a euclidean Lie algebra. Furthermore an indecomposable lorentzian Lie
algebra is either $\fsl(2,R)$ with a positive multiple of the Killing
form or a double extension of an even-dimensional euclidean abelian
Lie algebra by a one-dimensional Lie algebra acting via skew-symmetric
derivations.  As we will recall below, these are necessarily
generalised Nappi--Witten Lie algebras.

With these prefatory remarks we can now define the notions of
bargmannian, galilean and carrollian Lie algebras, as used in this
paper.

\begin{definition}
  A lorentzian metric Lie algebra $\g$ is said to be
  \emph{bargmannian} if it admits a (nonzero) central null element $Z
  \in \g$.
\end{definition}
Since $\fsl(2,\RR)$ has trivial centre, it follows that any
indecomposable bargmannian Lie algebra is a double extension of an
even-dimensional euclidean abelian Lie algebra by a one-dimensional
Lie algebra.  A general bargmannian Lie algebra is then the orthogonal
direct sum of an indecomposable bargmannian Lie algebra with a
euclidean Lie algebra (see \cite[Theorem~5]{Figueroa-OFarrill:2022pus}
for a more precise statement).

\begin{definition}
  A Lie algebra $\g$ is said to be \emph{carrollian} if it admits a
  (nonzero) central element $Z \in \g$ and a positive-semidefinite
  ad-invariant symmetric bilinear form $\left<-,-\right>$ whose
  radical is one-dimensional and spanned by $Z$.
\end{definition}

As shown in \cite[Prop.~2]{Figueroa-OFarrill:2022pus}, a carrollian
Lie algebra $\g_{\mathrm{car}}$ is isomorphic to a central extension of a
positive-definite metric Lie algebra $(\g_0,\left<-,-\right>_0)$; that
is, we have a (typically non-split) short exact sequence
\begin{equation}
  \begin{tikzcd}
    0 \arrow[r] & \RR \arrow[r] & \g_{\mathrm{car}} \arrow[r] & \g_0  \arrow[r] & 0.
  \end{tikzcd}
\end{equation}

\begin{definition}
  A Lie algebra $\g$ is said to be \emph{galilean} if it admits a
  (nonzero) ad-invariant $\tau \in \g^*$ and an ad-invariant $\gamma
  \in \odot^2\g$ which, as a symmetric bilinear form on
  $\g^*$, is positive-semidefinite and has a one-dimensional radical
  spanned by $\tau$.
\end{definition}

As shown in \cite[Prop.~3]{Figueroa-OFarrill:2022pus}, a galilean Lie
algebra $\g_{\mathrm{gal}}$ is isomorphic to an ``extension by
skew-symmetric derivation'' of a positive-definite metric Lie algebra
$(\g_0,\left<-,-\right>_0)$; that is, we have a short exact sequence
\begin{equation}
  \begin{tikzcd}
    0 \arrow[r] & \g_0 \arrow[r] & \g_{\mathrm{car}} \arrow[r] & \RR \arrow[r] & 0,
  \end{tikzcd}
\end{equation}
where the action of $\RR$ on $\g_0$ is via skew-symmetric derivations.

Taken together, these results may be summarised in a commutative
diagram of Lie algebras
\begin{equation} \label{eq:galcarbarg rels diagram}
    \begin{tikzcd}
        &            & 0\arrow[d]                     & 0 \arrow[d]                    &            & \\
        &            & \mathbb{R} \arrow[r,equal]\arrow[d]\ & \mathbb{R} \arrow[d] &            & \\
        & 0\arrow[r] & \g_{\mathrm{car}} \arrow[r]\arrow[d] & \g_{\mathrm{barg}} \arrow[r]\arrow[d] & \mathbb{R} \arrow[r]\arrow[d,equal] &0 \\
        & 0\arrow[r] & \g_0 \arrow[r]\arrow[d] &  \g_{\mathrm{gal}}\arrow[r]\arrow[d] & \mathbb{R}\arrow[r] &0 \\
        &                    & 0                     & 0                    &                       &
    \end{tikzcd}
\end{equation}
where $(\g_0,\left<-,-\right>_0)$ is a euclidean Lie algebra and
$\g_{\mathrm{barg}}$ is a bargmannian Lie algebra obtained by
double-extending $(\g_0,\left<-,-\right>_0)$ by the one-dimensional
Lie algebra spanned by a skew-symmetric derivation $D_0$.  The
underlying vector space of $\g_{\mathrm{barg}}$ is $\g_0 \oplus \RR D
\oplus \RR Z$ with Lie bracket given by
\begin{equation}
  \begin{split}
    [X,Y] &= [X,Y]_0 + \left<D_0 X, Y\right>_0 Z\\
    [D,X] &= D_0X\\
    [Z,-] &= 0,
  \end{split}
\end{equation}
for all $X,Y \in \g_0$, with $[-,-]_0$ the Lie bracket on
$\g_0$.  Via an automorphism of $\g_{\mathrm{barg}}$, we may bring the
inner product $\left<-,-\right>$ to the form
\begin{equation}
  \left<X,Y\right> = \left<X,Y\right>_0 \qquad\text{and}\qquad \left<D,Z\right> = 1,
\end{equation}
for all $X,Y \in \g_0$ and with all other inner products zero.  The
carrollian Lie algebra $\g_{\mathrm{car}}$ is the ideal $Z^\perp \subset
\g_{\mathrm{barg}}$ and the galilean Lie algebra is the quotient
$\g_{\mathrm{gal}} = \g_{\mathrm{barg}}/\RR Z$ and they inherit their
respective structures from those of $\g_{\mathrm{barg}}$.

Indecomposable bargmannian Lie algebras are generalised Nappi--Witten
Lie algebras $\nw_{2\ell + 2}$, defined explicitly as follows.
Consider a symplectic vector space $V$ spanned by
$P_1,\dots,P_{2\ell}$ with symplectic structure $\omega \in \wedge^2
V^*$ with $\omega(P_i,P_j):=\omega_{ij}$.  Then the underlying vector
space of $\nw_{2\ell+2}$ is $V \oplus \RR D \oplus \RR Z$, with
nonzero Lie brackets given by
\begin{equation}
  \label{eq:gen-NW-brackets}
  [P_i, P_j] = \omega_{ij} Z \qquad\text{and}\qquad [D,P_i] =
  \omega_{ij} P_j,
\end{equation}
and with the ad-invariant lorentzian inner product with nonzero entries
\begin{equation}
  \label{eq:gen-NW-metric}
  \left<P_i,P_j\right> = \delta_{ij} \qquad\text{and}\qquad \left<D,Z\right> = 1.
\end{equation}

\begin{definition}\label{def:gen nw alg}
  The Lie algebra $\nw_{2\ell+2}$ is a \emph{generalised Nappi--Witten algebra}.
\end{definition}

The justification for the name is that the special case $\ell=1$ was
used by Nappi and Witten to build the first WZW model based on a
non-abelian non-semisimple Lie group \cite{Nappi:1993ie}.

\subsection{WZW Models}
\label{sec:wzw-models}

A WZW model on a Lie group $G$ admitting a bi-invariant metric is a
two-dimensional field theory described by the action
\begin{equation}
  \label{eq:Standard WZW Model Action}
  S[g]=\int_{\Sigma} \langle g^{-1}\partial g, g^{-1} \bar{\partial} g \rangle + \int_{\Tilde{\Sigma}} \Tilde{g}^*\omega,
\end{equation}
where $\Sigma$ is an orientable Riemann surface, $g\colon\Sigma
\rightarrow G$ is a smooth map, $\Tilde{\Sigma}$ is some smooth
manifold with boundary $\partial \Tilde{\Sigma} = \Sigma$ and
$\Tilde{g} \colon \Tilde{\Sigma} \to G$ extends $g$.  In addition,
$g^{-1}\partial g$ and $g^{-1} \bar{\partial} g$ are the $(1,0)$- and
$(0,1)$- components, respectively, of the pull-back $g^* \vartheta$ of
the left-invariant Maurer--Cartan one-form on $G$ and $\omega$ is the
Cartan $3$-form (see below) on $G$.  Finally, $\langle
-,-\rangle$ stands both for the wedge product of one-forms and the
ad-invariant inner product on $\g$.  The extension of $g\colon
\Sigma \to G$ to $\Tilde{g} \colon \Tilde{\Sigma} \to G$ is
topologically unobstructed for any Lie group $G$, because of the fact
that the second homotopy group of any Lie group is trivial.

The first term in \eqref{eq:Standard WZW Model Action} is known as
the \emph{principal chiral model (PCM)} whereas the second term is known
as the \emph{Wess--Zumino (WZ)} term, which is the integral over
$\Tilde{\Sigma}$ of the pull-back of the Cartan 3-form $\omega$
on $G$. In terms of the (left-invariant) Maurer--Cartan one-form
$\vartheta\in\Omega^1(G;\g)$,
\begin{equation}
  \omega = \tfrac{1}{6}\langle \vartheta, [\vartheta,\vartheta]\rangle
  = -\tfrac{1}{3}\langle\vartheta,d\vartheta\rangle,
\end{equation}
where we have used the structure equation
\begin{equation} \label{eq:structure equation MC form}
  d\vartheta + \tfrac{1}{2} [\vartheta, \vartheta] = 0.
\end{equation}
Hence, we may write \eqref{eq:Standard WZW Model Action} in its more
familiar form (using $g^*\vartheta = g^{-1}dg$ for any matrix Lie
group G)
\begin{equation}
  \label{eq:Standard WZW Action Explicit}
  S[g]=\int_{\Sigma} \left<g^{-1}\partial g, g^{-1} \bar{\partial} g
  \right> + \tfrac{1}{6} \int_{\Tilde{\Sigma}} \left<
    \Tilde{g}^{-1}d\Tilde{g}, [\Tilde{g}^{-1}d\Tilde{g}, \Tilde{g}^{-1}d\Tilde{g}] \right>.
\end{equation}
At face value, the WZ term seems to depend on  $\Tilde{g} \colon
\Tilde{\Sigma} \to G$, but its variation results in local field
equations for the field $g \colon \Sigma \to G$.  The quantum theory does
not depend on $\Tilde{\Sigma}$ either, provided one normalises the
inner product in such a way that exponential of the action in the path
integral changes by an integer multiple of $2\pi$
\cite{Witten:1983ar}.

The WZW action $S[g]$ is invariant under the following transformations:
\begin{itemize}
\item conformal transformations
\item the isometry group of $G$ with its bi-invariant metric.  This is
  typically $G\times G$, but since the left and right multiplication
  by the centre coincide, when $G$ has nontrivial centre, we must
  quotient by the centre, resulting in $(G \times G)/Z$, with $Z
  \subset G$ the centre.
\item the infinite-dimensional group $G(z)\times G(\Bar{z})$, acts by
  $g(z,\bar{z})\mapsto\Omega(z)g(z,\bar{z})\bar{\Omega}(\bar{z})^{-1}$,
  where again constant $\Omega$ and $\bar{\Omega}$ in the centre act
  in the same way, so we should quotient by $Z$, resulting in
  $\left(G(z)\times G(\Bar{z})\right)/Z$.
\end{itemize}
The final "enhanced" symmetry only exists for our choice of
normalisation of the WZW action. The currents associated to the
$G(z)\times G(\Bar{z})$ symmetry satisfy an affine version of the
Lie algebra of G, which is then encapsulated by the operator product
expansions (OPEs) of these currents in the corresponding quantum
CFT. This lays the bedrock of how WZW models are generally studied.

\section{Gauged WZW models}
\label{sec:gauged-wzw-models}

The study of WZW models with some subgroup of its isometry group being
gauged has been of interest for various reasons, but most notably due
to how this procedure gives rise to coset CFTs as in the GKO
construction \cite{Goddard:1986ee}.

However, one cannot gauge any subgroup of the isometry group of the
WZW model. Additionally, the type of gauging we perform in this paper
does not lead to a coset construction. We address these points in this
section, along with an explicit discussion of the gauging procedure
for our purposes.

Throughout this section, $G$ refers to a generic Lie group and $\B$ refers to 
an indecomposable bargmannian Lie group with Lie algebra $\nw_{2\ell+2}$.

\subsection{What can we gauge?}

There are certain (cohomological) restrictions
\cite{Hull:1989jk, Witten:1991mm, Figueroa-OFarrill:1994uwr} as to what
subgroups of the isometry group $G\times G$ we can gauge, summarised in the 
following proposition.
\begin{proposition}
  \label{prop:gauge-able subgroups}
  Let $G$ be a Lie group with a bi-invariant metric and $\g$ its Lie
  algebra.  Let $\left<-,-\right>$ denote the ad-invariant inner
  product on $\g$.  Let $H \subset G \times G$ be a Lie subgroup with
  Lie algebra $\h$. The embedding $\h \subset \g \times \g$ defines
  two Lie algebra homomorphisms $l,r\colon \h \rightarrow \g$ by
  composing with the cartesian projections.  Then $H$ can be gauged if
  and only if
  \begin{equation*}
    l^*\langle-,-\rangle=r^*\langle-,-\rangle;
  \end{equation*}
  equivalently,
  \begin{equation*}
    \langle l(X),l(Y) \rangle = \langle r(X),r(Y) \rangle \quad \forall X,Y \in \h.
  \end{equation*}
  Such subgroups are said to be \emph{anomaly-free}.
\end{proposition}

The two most obvious anomaly-free subgroups are:
\begin{enumerate}
\item $l=r\neq 0$, which corresponds to \emph{diagonal gauging} and
  give rise to the coset models mentioned earlier;
\item and $l=0$ \textit{or} $r=0$, which corresponds to \emph{null chiral gauging}.
\end{enumerate}

For the WZW model on $\B$ with Lie algebra
$\b=\nw_{2\ell+2}$, where the subgroup $H$ is one-dimensional and
generated by a null element $Z \in \g$, any embedding
$\h \in Z \mapsto (\alpha Z, \beta Z) \in \g \oplus \g$, for
$\alpha,\beta \in \RR$, satisfies the condition in
Proposition~\ref{prop:gauge-able subgroups}.  Indeed,
$\left<l(Z),l(Z)\right> = \left<\alpha Z, \alpha Z\right> = \alpha^2
\left<Z,Z\right> = 0$ and, similarly,
$\left<r(Z),r(Z)\right> = \beta^2 \left<Z,Z\right> = 0$.  But,
moreover, since $Z$ is also central, it turns out that for any $\alpha
\neq \beta$ the gaugings are equivalent.  Without loss of generality
we will take $\alpha = 1$ and $\beta = 0$ in our treatment.

\subsection{How do we gauge?}

Consider a subgroup $H\subset G\times G$ that obeys
Proposition~\ref{prop:gauge-able subgroups}. 
We treat the null chiral gauging of each of the two terms in
\eqref{eq:Standard WZW Model Action} separately.  The PCM can be
gauged straightforwardly by minimal coupling; this involves modifying
the exterior derivative as $d\to d + \mathbf{A}$, where
$\mathbf{A}=A+\bar{A}$ is a gauge field that can be interpreted as a
one-form on $G$ with values in $\h^\CC$.

The WZ term is not as straightforward to gauge. There are certain
obstructions to gauging \cite{Hull:1989jk} that are encapsulated by
equivariant cohomology. The key result from
\cite{Witten:1991mm,Figueroa-OFarrill:1995lva} (also see
\cite{Figueroa-OFarrill:2005vws}) is that the WZ term can be gauged if
and only if one can extend $\omega\in \Omega^3(G)$ to an equivariant
closed form ${\varOmega}$. It turns out that using an algebraic model
of equivariant cohomology, known as the Cartan model, makes it
feasible to perform this procedure explicitly. We demonstrate this
procedure here, but we refer the reader to
\cite{Figueroa-OFarrill:2005vws} for more details.

Let $\{X_a\}_{a=1}^{\dim H}$ be a basis for the Lie algebra $\h$ of
$H$. The Lie bracket can be written in terms of this basis as $[X_a,
X_b]= f\indices{_a_b^c}X_c$. The action of $H$ on $G$ is generated by
the corresponding left-invariant Killing vector fields
\begin{equation*}
  \xi_a(g) = \frac{d}{dt}\Big|_{t=0} g\cdot\exp(tX_a).
\end{equation*}
Define $\iota_a$ and $\mathcal{L}_a$ to be the contraction and Lie derivative on $G$ with respect to $\xi_a$.
Consider the cochain complex of $H$-invariants
\begin{equation*}
  \left\{\big(\Omega(G)\otimes\mathfrak{S}^{\dotr}\h^*\big)^H=\bigoplus_{p=0}^{n} C^p_H,d_C\right\}.
\end{equation*}
Elements of $C^p_H$ are linear combinations of 
\begin{equation*}
    \Phi = \phi + \theta_a F^a + \tfrac{1}{2}\varphi_{ab} F^a F^b + \dots,
\end{equation*}
where $\theta_a\in \Omega^{p-2}(G)$, $\phi_{ab}\in\Omega^{p-4}(G)$ and
so on, and $F=F^a X_a$ is an abstract $\h$-valued
2-form.
Note that in this algebraic model for equivariant cohomology,
$F$ is simply a degree-2 object that is not a 2-form on some specific
manifold.
$H$-invariance implies that $\mathcal{L}_a\theta_b = f\indices{_a_b^c}\theta_c$, etc. The differential $d_C\colon C^p_H \rightarrow C^{p+1}_H$ acts as
\begin{equation} \label{eq:d_C action}
    d_C (\phi) = d\phi - F^a\wedge \iota_a \phi, \quad d_C F^a = 0
\end{equation}
for all $\phi \in \Omega^{\dotr}(G)$, where $d$ is the de Rham differential. The problem of gauging the WZ term by a subgroup $H$ is equivalent to that of extending the Cartan 3-form $\omega\in \Omega^3(G)$ to $\varOmega_C  \in C^3_H$ such that $d_C \varOmega_C = 0$. Explicitly, the Cartan representative $\Omega_C$ is of the form 
\begin{equation} \label{eq:gauged Omega Cartan REP}
    \varOmega_C = \omega + \theta_a F^a.
\end{equation}
Using \eqref{eq:d_C action}, demanding $d_C \varOmega = 0$ means
\begin{equation}
    d_C (\omega + \theta_a F^a) = d_C \omega + (d_C \theta_a) F^a = 0.
\end{equation}
But $d \omega = 0$ (since it is the Maurer--Cartan 3-form), so $d_C\omega = - F^a \iota_a \omega$. Writing $d_C \theta_a = d\theta_a - \iota_b \theta_a F^b$ and grouping terms according to powers of $F^a$, we see that $d_C \varOmega_C = 0$ if and only if
\begin{gather}
    \iota_a \omega = d\theta_a \label{eq:WZ Term gauge condition 1} \\
    \iota_a \theta_b = -\iota_b \theta_a \label{eq:WZ Term gauge condition 2}.
\end{gather}
To obtain an explicit expression for the gauged WZ term from this algebraic model, we need to minimally couple to an abstract 1-form $A$, which will play the role of the gauge field in the WZW model. This requires demanding the relation\footnote{Strictly speaking, this $d$ should refer to the twisted Cartan differential $d_C$. The minimal coupling map \eqref{eq:F=dA minimal coupling} modifies $d_C$ on $F\in\h^*$. This map is a quasi-isomorphism of chain complexes. However, since we immediately revert to interpreting $A$ as a $\h^\CC$-valued 1-form on $\Sigma$, we simply use $d$.}
\begin{equation}\label{eq:F=dA minimal coupling}
    F = F^a X_a = dA +\tfrac{1}{2}[A,A] = (dA^a + \tfrac{1}{2}f\indices{_b_c^a}A^b A^c)X_a
\end{equation}
Finally, equation \eqref{eq:F=dA minimal coupling}, along with the conditions \eqref{eq:WZ Term gauge condition 1} and \eqref{eq:WZ Term gauge condition 2} which let us to determine $\theta_a$, can be used together in \eqref{eq:gauged Omega Cartan REP} to complete the gauging of the WZ term.

We now apply these results to our case, in which $G=\B$ is an
indecomposable bargmannian Lie group with Lie algebra
$\b=\nw_{2\ell+2}$. The WZW model on the corresponding group $\B$,
which we call a \emph{generalised Nappi--Witten group},
is\footnote{Later, we will use angle brackets to denote dual-pairing
  as well, so the subscript has been introduced here to make it clear
  that these angle brackets refer to the inner product on the Lie
  algebra $\b$.}:
 \begin{equation}
 \label{eq:bargmannian WZW action}
     S[g]=\int_{\Sigma} \langle g^{-1}\partial g, g^{-1} \bar{\partial} g \rangle_{\mathfrak{b}} -\tfrac{1}{3} \int_{\Tilde{\Sigma}} \langle \Tilde{g}^{-1} d\Tilde{g}, d(\Tilde{g}^{-1} d\Tilde{g}) \rangle_{\mathfrak{b}}.
 \end{equation}
We want to promote the isotropic subgroup $H\subset
\B \times \B$ to a local symmetry. Recalling that $H$ embeds into $\B\times \B$ via $H \mapsto (H,e)$, the goal is to find an an action that is invariant under
\begin{equation}
  g(z,\Bar{z}) \mapsto h(z,\Bar{z}) g(z,\Bar{z}).
\end{equation}
It is clear from the fact that $H \subset \B$ is a central subgroup
that the RHS equals $g(z,\Bar{z}) h(z,\Bar{z})$, or indeed any mixture
of left and right multiplications.  This simply reiterates the
independence of the resulting gauged WZW model on the precise
embedding $H \subset \B \times \B$.

Gauging the first term in \eqref{eq:bargmannian WZW action} gives
\begin{equation} \label{eq:PCM gauged}
  \int_{\Sigma} \langle g^{-1}\partial g, g^{-1} \bar{\partial} g \rangle_{\mathfrak{b}} \to 
  \int_{\Sigma} \langle g^{-1}\partial g, g^{-1} \bar{\partial} g \rangle_{\mathfrak{b}} + 
  \int_{\Sigma} \langle A, g^{-1} \bar{\partial} g \rangle_{\mathfrak{b}} +
  \int_{\Sigma} \langle g^{-1}\partial g, \Bar{A} \rangle_{\mathfrak{b}},
\end{equation}
where the $\left<A,\Bar{A}\right>$ term is absent since $A$ and
$\Bar{A}$ are $\h$-valued and $\h$ is an isotropic subalgebra.

To gauge the second term in \eqref{eq:bargmannian WZW action}, note that $\dim H = 1$, $\h = \RR Z$ and the corresponding vector field $\xi_Z$ is both left and right-invariant, by virtue of $Z$ being central. 
For one-dimensional $\h$, the conditions \eqref{eq:WZ Term gauge condition 2} and \eqref{eq:WZ Term gauge condition 1} give
\begin{gather}
    \iota_{Z} \theta_Z = 0 \label{eq:contraction of v_Z is zero} \quad \text{and} \\
    \iota_Z \omega = \tfrac{1}{6}\iota_Z \langle \vartheta, [\vartheta, \vartheta]\rangle_{\mathfrak{b}} = \tfrac{1}{2}\langle \iota_Z \vartheta, [\vartheta, \vartheta]\rangle_{\mathfrak{b}} = d \theta_Z
\end{gather}
respectively.
Using $\iota_Z \vartheta = Z$ (contraction of the left-invariant Maurer--Cartan form with a left-invariant vector field) and the fact that $Z\in[\b,\b]^\perp$ due to $Z$ being central, we get
\begin{equation} \label{eq:deciphering theta_Z}
    \iota_Z \omega = \tfrac{1}{2}\langle Z, [\vartheta, \vartheta] \rangle_{\mathfrak{b}} = 0 = d\theta_Z.
\end{equation}
The Maurer--Cartan form $\vartheta\in \Omega^1(\B;\b)$ obeys the structure equation \eqref{eq:structure equation MC form},
so equation \eqref{eq:deciphering theta_Z} tells us that
\begin{equation} \label{eq:theta_Z equations}
    d\theta_Z =0, \quad d(\theta_Z + \langle \vartheta, Z \rangle_{\mathfrak{b}}) = 0.
\end{equation}
Hence, we may conclude that 
\begin{equation}
    \theta_Z = -\langle \vartheta, Z \rangle_{\mathfrak{b}} + \pi^*\tau.
\end{equation}
Here, $\pi \colon \B \rightarrow \B/H$ is the projection map and
$\tau\in\Omega^1(\B/H)$ is closed. Thus, the gauging procedure of the
WZ term is ambiguous up to a closed 1-form on $\B/H$.  It might be
interesting to study how this ambiguity can be exploited in the
resulting gauged WZW model, but for convenience, we will choose $\tau
= 0$ for the rest of this paper and hence
\begin{equation}
    \label{eq:theta_Z final}
    \theta_Z = - \langle \vartheta, Z \rangle_{\mathfrak{b}}.
\end{equation}
In our case, $F = dA \iff F^Z = dA^Z$ appears in \eqref{eq:gauged Omega Cartan REP}. Hence, using \eqref{eq:theta_Z equations} and \eqref{eq:theta_Z final}, we deduce that
\begin{equation}
    \Omega_C = \omega + dA^Z \theta_Z = \omega - d(A^Z\wedge \langle \vartheta,Z \rangle_\mathfrak{b}).
\end{equation}
Substituting for $\Omega_C$ in the gauged WZ term gives
\begin{equation}
    \int_{\Tilde{\Sigma}} \Tilde{g}^*\Omega_C 
    = \int_{\Tilde{\Sigma}} \Tilde{g}^*\omega - \int_{\Tilde{\Sigma}} \Tilde{g}^*d(A^Z\wedge \langle \vartheta,Z \rangle_\mathfrak{b})
    = \int_{\Tilde{\Sigma}} \Tilde{g}^*\omega - \int_{\Sigma}(A^Z\wedge \langle g^*\vartheta,Z \rangle_\mathfrak{b}),
\end{equation}
where we have used the fact that the de Rham differential commutes with pull-backs and Stokes' theorem on the second term. Finally, identifying $A^Z Z$
as the $\h^\CC$-valued 1-form $\mathbf{A}$ on $\Sigma$ which we had
defined earlier and $g^*\vartheta = g^{-1}dg = g^{-1} \partial g +
g^{-1} \Bar{\partial} g$, we arrive at the following expression for
the gauged WZ term:
\begin{equation}
    \label{eq:Gauged WZ Term}
     \int_{\Tilde{\Sigma}} \Tilde{g}^*\Omega_C = \int_{\Tilde{\Sigma}} \Tilde{g}^*\omega - \int_{\Sigma} \langle A, g^{-1} \Bar{\partial} g \rangle_\mathfrak{b} - \int_{\Sigma} \langle \Bar{A}, g^{-1} \partial g \rangle_\mathfrak{b}.
\end{equation}
 Putting the two terms \eqref{eq:PCM gauged} and \eqref{eq:Gauged WZ Term} together, we arrive at the gauged WZW action of our theory:
\begin{equation}
    \label{eq:Gauged WZW Action}
    S[g,\Bar{A}] = S[g] - 
     2\int_{\Sigma} \langle \Bar{A}, g^{-1} \partial g \rangle_{\mathfrak{b}}.
\end{equation}
Notice that the holomorphic gauge field $A$ vanishes under this
gauging procedure -- an artefact of the gauging being chiral.

\section{The BRST treatment}\label{sec:Implement gauge constraint via BRST}

This section presents the step-by-step implementation of the null
chiral gauging constraint via the BRST formalism for 2d CFTs.  To do
this, we need to bring our theory to the desired form by quantising
the gauged action \eqref{eq:Gauged WZW Action} using path integrals,
making use of standard manipulations such as those in
\cite{Karabali:1988au}, adapted to null chiral gauging. Starting with
\begin{equation}
  \mathcal{Z} = \int \mathcal{D}g \mathcal{D}\Bar{A} \exp{\left(-S[g,\Bar{A}]\right)},
\end{equation}
we choose the holomorphic gauge $\Bar{A}=0$. This gives rise to the Faddeev--Popov determinant
\begin{equation}
  \mathcal{Z} = \int \mathcal{D}g (\det\Bar{\partial}) \exp{\left(-S[g]\right)},
\end{equation}
where \cite{Karabali:1988au}
\begin{equation}
  \det \Bar{\partial} = \int\mathcal{D}B\mathcal{D}C \exp{\left(-\int_\Sigma \langle B, \Bar{\partial} C \rangle\right)}.
\end{equation}
We are therefore left with
\begin{equation}
  \label{eq:path integral final partition function}
  \mathcal{Z} = \int\mathcal{D}g\mathcal{D}B\mathcal{D}C \exp{\left(-S[g]-\int_\Sigma \langle B, \Bar{\partial} C \rangle\right)}.
\end{equation}
 
The expression in \eqref{eq:path integral final partition function}
tells us that the holomorphic sector of the resulting quantum field
theory consists of two CFT sectors: a WZW model on the bargmannian
group $\B$ and a weight-$(1,0)$ $bc$-system, coupled by a constraint
arising from wanting to gauge the isotropic subgroup $H\subset G$,
which we may implement in the BRST formalism. In particular, the
constraint is implemented by the square zero BRST operator, whose
explicit form we derive in Section~\ref{sec:setting-up-cfts}
below. The resulting gauged theory is the cohomology of this
operator. Since each sector in \eqref{eq:path integral final partition
  function} is a CFT, we exploit the computational ease and power of
OPEs.

Before we proceed, it should be noted that the resulting quantum field
theory has different holomorphic and anti-holomorphic sectors due to
the vanishing of $\Bar{A}$ mentioned earlier. Intuitively, this seems
to be an expected feature of a chirally gauged theory.

\subsection{Setting up the CFTs in the gauged WZW action}
\label{sec:setting-up-cfts}

The first sector in \eqref{eq:path integral final partition function} is the WZW model on $\B$ with Lie algebra $\mathfrak{b}=\nw_{2\ell+2}$. Recall that we choose to work with the symplectic form given in \eqref{eq:symplectic form omega}. Its CFT is described by the weight-1 currents
$\{P_1,\dots, P_{2\ell}, D, Z\}$ with OPEs
\begin{equation}
  \label{eq:General NW Currents OPEs}
  A(z)B(w) = \frac{\langle A,B \rangle_{\mathfrak{b}}(w)}{(z-w)^2} + \frac{[A,B](w)}{(z-w)^2}+\reg.
\end{equation}
for all $A,B\in\{P_1,\dots, P_{2\ell}, D, Z\}$. There exists a
two-parameter family of non-degenerate inner products $\langle-,-\rangle_{\mathfrak{b}}$ on $\mathfrak{b}$ given by 
\begin{equation} \label{eq:2-param family of metrics}
  \langle P_i, P_j \rangle_\mathfrak{b} = \delta_{ij}, \quad \langle D,D \rangle_\mathfrak{b} = \kappa_1 \quad \langle Z,Z \rangle_\mathfrak{b} = 0,\quad \langle D, Z \rangle_\mathfrak{b} = \kappa_2,
\end{equation}
where $\kappa_1,\kappa_2\in\RR$. Rescaling $Z\to Z/\kappa_2$ and
performing the Lie algebra automorphism $D\to D - {\kappa_1}/{2}\, Z
$ lets us set $\langle D,Z \rangle = 1$ and $\langle D,D \rangle =
0$.  From now on, we will work with this particular choice of inner product
on $\mathfrak{b}$.  The Lie bracket is is given by
equation~\eqref{eq:gen-NW-brackets}; that is,
\begin{equation}
  [P_i, P_j] = \omega_{ij} Z, \quad\quad [D, P_i] = \omega_{ij} P_j.
\end{equation}
Substituting the above inner product and Lie bracket into \eqref{eq:General NW Currents OPEs}, we obtain the following OPEs:
\begin{align}
  P_i(z) P_j(w) &= \frac{\delta_{ij}\mathbbm{1}(w)}{(z-w)^2} + \frac{\omega_{ij} Z(w)}{z-w} + \reg. \\
  D(z) P_i(w) &= \frac{\omega_{ij} P_j(w)}{z-w} + \reg. \\
  D(z) Z(w) &= \frac{\mathbbm{1}(w)}{(z-w)^2} + \reg.
\end{align}
For the ease of calculations and constructing representations, we prefer to work with the complexification $\mathfrak{b}^\CC$, and perform the variable change
\begin{equation}
  P^{\pm}_{a} \defeq \tfrac{1}{\sqrt{2}}\big(P_{2a-1} \pm iP_{2a}\big), \quad J\defeq iD, \quad I\defeq -iZ,
\end{equation}
where $a\in\{1,\dots,\ell\}$.
Note that this variable change is an isomorphism of complex Lie
algebras, but not real Lie algebras. In other words, this variable
change results in a choice of different real form for $\b$, one whose
inner product has split signature $(2,2)$ instead of $(3,1)$. We
proceed with this caveat in mind and write the new brackets as
\begin{equation} \label{eq:generalised_NW_complexified}
  [P^+_a, P^-_b] = \delta_{ab} I, \quad [J, P^{\pm}_a] = \pm P^{\pm}_a
\end{equation}
and where the inner product is
\begin{equation} \label{eq:new (2,2) metric}
  \langle P^+_a, P^-_b\rangle = \delta_{ab}, \quad \langle I,J \rangle = 1
\end{equation}
and zero otherwise. These translate into the OPEs of the corresponding currents in the CFT
\begin{align}
  P^+_a (z) P^-_b (w) &= \frac{\delta_{ab}\mathbbm{1}(w)}{(z-w)^2} + \frac{\delta_{ab} I(w)}{z-w} + \reg. \label{eq:gen NW OPE for P+P-}\\
  J(z) P^{\pm}_a(w) &= \frac{\pm P^{\pm}_a(w)}{z-w} + \reg. \label{eq:gen NW OPE for J P+-} \\
  J(z) I(w) &= \frac{\mathbbm{1}(w)}{(z-w)^2} + \reg. \label{eq:gen NW OPE for JI}
\end{align}

The second sector in \eqref{eq:path integral final partition function}
are the Faddeev--Popov ghosts: a weight-$(1,0)$ fermionic $bc$-system
$(B,C)$. The ghost fields $B(z)$ and $C(z)$ have the usual OPEs
\begin{equation}
  \label{eq:(1,0) BC-system OPEs}
  B(z) C(w) = \frac{\mathbbm{1}(w)}{z-w}.
\end{equation}
They are free fields of the CFT with stress tensor 
\begin{equation}
  T^{BC} = -\normalord{B\partial C}
\end{equation}
of central charge $-2$. Once again, there is an identical ``barred''
weight-$(1,0)$ $bc$-system given by the fields $(\bar{B}, \bar{C})$
and Virasoro element $\bar{T}^{BC}$.

Finally, we need to deduce the form of the BRST operator. This means
we need to deduce what constraints are imposed by gauging the isometry
subgroup generated by $Z\in\nw_{2\ell+2}$. To do this, let us perform the
variation of the gauge field in the gauged WZW action \eqref{eq:Gauged
  WZW Action}. Note that the second term of \eqref{eq:Gauged WZW
  Action} can be written as
\begin{equation}
  \int_{\Sigma} \langle \Bar{A}, g^{-1} \partial g \rangle_\mathfrak{b} = \int_{\Sigma} \langle \Bar{A}, \partial g g^{-1} \rangle_\mathfrak{b}.
\end{equation}
due to $\Ad$-invariance of the inner product and the fact that
$\Ad_g Z = Z$ for all $g\in\B$. Hence, performing the variation
$\delta \bar{A}$ in \eqref{eq:Gauged WZW Action} constrains the
holomorphic current $\partial g g^{-1}$ along the $J$ direction to be
zero. Quantum mechanically, this is a first-class constraint given by
$J(z)=0$, which is implemented by the BRST current
\begin{equation}
  \label{eq:j_NG original}
  j_{\mathrm{NG}} = \normalord{CJ}.
\end{equation}
A quick computation shows that $d_{\mathrm{NG}}^2 = 0$, reiterating the
fact that we have gauged an anomaly-free subgroup
\cite{Figueroa-OFarrill:1995lva}. Another key point to remember is
that this constraint is only being implemented in the holomorphic
sector (i.e., there is no ``barred'' null gauging BRST operator). This
will be important when we build the full string theory in
Section~\ref{sec:full_theory}.

\subsection{Free field realisation} \label{sec:FFR}

 Having a free field realisation of the CFT would be particularly
 useful to compute BRST cohomology explicitly. From this point on, we
 specialise to $\ell=1$, in which case $\nw_4^\CC$ is the complexified
 Nappi--Witten algebra. The Wakimoto-type embedding for such a
 CFT\footnote{The anti-holomorphic part is identical, with the
   corresponding fields denoted with bars.} into the CFT of two
 weight-$(1,0)$ bosonic $\beta\gamma$-systems is given below. This was
 inspired by the result in \cite[Theorem~5.4]{bao2011representations} (to which we return in Section \ref{sec:FFR_comments}):
\begin{equation}
\label{eq:Wakimoto embedding}
    \begin{split}
        &P^+ = \beta\\
        &P^- = \partial \gamma + \tilde{\beta} \gamma \\
        &I = \tilde{\beta}\\
        &J = \partial\tilde{\gamma} + \tfrac{1}{2}\tilde{\beta} - \normalord{\beta\gamma}, \\
    \end{split}
\end{equation} 
where the $\beta\gamma$-systems have the usual non-zero OPEs
\begin{align}
    \beta(z) \gamma(w) &= \frac{\mathbbm{1}(w)}{z-w} + \reg \label{eq:beta-gamma OPE} \\
    \tilde{\beta}(z) \tilde{\gamma}(w) &= \frac{\mathbbm{1}(w)}{z-w} + \reg
\end{align}
Computation of the OPEs of $P^{\pm}$, $I$ and $J$ with the expressions
in \eqref{eq:Wakimoto embedding} indeed gives \eqref{eq:gen NW OPE for
  P+P-}, \eqref{eq:gen NW OPE for J P+-} and \eqref{eq:gen NW OPE for
  JI}.

We now make two remarks:
\begin{itemize}
\item The embedding~\eqref{eq:Wakimoto embedding} is not conformal.
  The Sugawara Virasoro element
  \begin{equation} \label{eq:Tsug} 
    \Tsug = \normalord{P^+ P^-} + \normalord{IJ} -\tfrac{1}{2}\partial I - \tfrac{1}{2}\normalord{II}
  \end{equation}
  corresponds to
  \begin{equation} \label{eq:Tsug embedding}
    \Tsug = \normalord{\beta\partial\gamma} + \normalord{\tilde{\beta}\partial\tilde{\gamma}} - \tfrac{1}{2}\partial \tilde{\beta}
  \end{equation}
  under the embedding.  The offending term proportional to
  $\partial\tilde{\beta}$ is responsible for the fact that
  $\tilde{\gamma}$ is no longer primary under $\Tsug$.  This will turn
  out not to affect the calculation of the Virasoro BRST cohomology in
  Section \ref{sec:full_theory}.

\item We have chosen a normalisation for the inner product on the
  Nappi--Witten algebra.  In some of the literature (see, e.g.,
  \cite{Babichenko_2021}), a notion of ``level'' is introduced for the
  affine Nappi--Witten algebra.  The level plays an important role in
  the representation theory of affine Lie algebras built out of
  semisimple Lie algebras, but it does not for the affine
  Nappi--Witten Lie algebra.  Indeed, suppose that we rescale the
  inner product by the ``level'' $k$, assumed nonzero. Then one
  may simply redefine $P_- \mapsto k^{-1} P_-$ and
  $I \mapsto k^{-1} I$, while keeping $P_+$ and $J$ the same.  Then
  the redefined fields satisfy the affine NW algebra at ``level'' 1.
\end{itemize}

\subsection{Null gauging cohomology} \label{sec:null_gauging_cohomology}

Having embedded the CFT of the WZW model on $\B$ into the CFT of two
weight-$(1,0)$ $\beta\gamma$-system, we proceed with computation of
null gauging BRST cohomology. First, we explicitly construct the space of states of our system. Recall that 
general weight-$(\lambda,1-\lambda)$ $bc$- and
$\beta\gamma$-systems, where $\lambda\in\ZZ$, admit mode expansions
\begin{gather}
    \beta(z) = \sum_n \beta_n z^{-n-1}, \ \ \gamma(z) = \sum_n \gamma_n z^{-n}\label{eq:beta-gamma mode exp general}\\
    B(z) = \sum_n B_n z^{-n-1}, \ \ C(z) = \sum_n C_n z^{-n}\label{eq:BC mode exp general}
\end{gather}
we can build the space of states $V$ by acting these modes on the
respective vacua.

Their vacuum states $\ket{\rho}_{BC}$ and $\ket{\sigma}_{\beta\gamma}$ with
charges $\rho,\sigma\in\ZZ$ satisfy
\begin{equation}\label{eq:bcbetagamma vacua}
\begin{alignedat}{2}
   &B_n\ket{\rho}_{BC} = 0  \quad \forall n\geq \rho+1-\lambda
   \quad\quad && \beta_n \ket{\sigma}_{\beta\gamma} = 0 \quad n\geq -\sigma+1-\lambda \\
   &C_n\ket{\rho}_{BC} = 0  \quad \forall n \geq -\rho+\lambda
   \quad\quad &&\gamma_n \ket{\sigma}_{\beta\gamma} = 0 \quad n\geq \sigma + \lambda.
\end{alignedat}
\end{equation}
The spaces of states $V^{BC}$ and $V^{\beta\gamma}_\sigma$ are linear
combinations of monomials built from the modes that act non-trivially
on the vacuum. Note in particular the explicit dependence on the
vacuum choice only in the bosonic case. This is because the $BC$ vacua
are related to one another through the action of the $BC$ mode
algebra, but there is no such relationships among the $\beta\gamma$
vacua. Thus, each choice of $\beta\gamma$ vacuum leads to a different
picture \cite{Friedan:1985ge}, which we present as a label on the
resulting space of states.

Coming back to our CFT under the embedding \eqref{eq:Wakimoto embedding}, the null gauging BRST current
\eqref{eq:j_NG original} is now
\begin{equation} \label{eq:j_NG embedded}
    j_{\mathrm{NG}} = \tfrac{1}{2}\normalord{C\betil} + \normalord{C\partial\gamtil}-\normalord{C \beta\gamma}.
\end{equation}
The BRST operator $d_{\mathrm{NG}}$ acts on the space of states
\begin{equation}
  V = V^{\beta\gamma}_{\sigma}\otimes V^{BC} \otimes V^{\betil\gamtil}_{\tilde{\sigma}},
\end{equation}
where $V^{\beta\gamma}_{\sigma}$, $V^{\betil\gamtil}_{\tilde{\sigma}}$ and $V^{BC}$
are the space of states the $(\beta,\gamma)$, $(\betil,\gamtil)$ and
$(B,C)$ systems respectively. Explicitly, setting $\lambda=1$ in \eqref{eq:bcbetagamma vacua} 
lets us infer that the creation operators are 
\begin{equation*}
  \{B_{n\leq \rho-1}, C_{n\leq -\rho}, \betil_{n\leq -\tilde{\sigma}-1}, \gamtil_{n\leq \tilde{\sigma}}, \beta_{n\leq -\sigma-1}, \gamma_{n\leq \sigma}\}.  
\end{equation*}
Hence, $V=V^{\beta\gamma}_{\sigma} \otimes V^{BC} \otimes
V^{\betil\gamtil}_{\tilde{\sigma}}$ is spanned by monomials of the form
\begin{equation}
\label{eq:monomials of full space Wakimoto}
\begin{split}
 \ket{\psi}&= B_{-n_1}\dots B_{-n_\mathcal{B}}C_{-m_1}\dots C_{-m_\mathcal{C}}\ket{\rho}_{BC} \\
 &\otimes 
 \betil_{-k_1}\dots \betil_{-k_{\tilde{\mathfrak{B}}}} \gamtil_{-l_1}\dots\gamtil_{-l_{\tilde{\mathfrak{C}}}} \ket{\tilde{\sigma}}_{\betil\gamtil} \\
 &\otimes
 \beta_{-r_1}\dots \beta_{-r_\mathfrak{B}} 
 \gamma_{-s_1}\dots\gamma_{-s_\mathfrak{C}}\ket{\sigma}_{BC},   
\end{split}
\end{equation}
where 
\begin{gather*}
n_1 > \dots > n_{\mathcal{B}}  \geq -\rho+1,\ \ m_1 > \dots > m_{\mathcal{C}} \geq \rho\\
k_1 \geq \dots \geq k_{\tilde{\mathfrak{B}}} \geq \tilde{\sigma}+1,\ \ l_1 \geq \dots \geq l_{\tilde{\mathfrak{C}}} \geq -\tilde{\sigma}\\
 r_1 \geq \dots \geq r_{\mathfrak{B}}\geq \sigma+1,\ \  s_1 \geq \dots \geq s_{\mathfrak{C}}\geq -\sigma.
\end{gather*}
We now wish to find the spectrum of the null chirally gauged
Nappi--Witten model, which is done by computing the cohomology
$H^{\dotr}(V)$ with respect to $d_{NG}$. We summarise the result in
the following proposition.
\begin{proposition} \label{prop:NG cohomology}
The cohomology with respect to $\dNG$ is
    \begin{equation}
        H^{\dotr}(V)\cong \CC \ket{\mathrm{vac}}_{\tilde{\sigma}}\otimes V^{\beta\gamma}_{\sigma},
    \end{equation}
    where $\ket{\mathrm{vac}}_{\tilde{\sigma}} =
    \ket{\tilde{\sigma}}_{\betil\gamtil} \otimes
    \ket{-\tilde{\sigma}}_{BC}$ is a choice of picture labelled by
    $m\in\ZZ$.  Hence, the spectrum of the holomorphic sector of the
    null chirally gauged Nappi--Witten model is described by a $\beta\gamma$-system.
\end{proposition}
At first glance, this drastic simplification seems surprising, but
this is due to the \emph{Kugo--Ojima (KO) quartet mechanism}. Details
of this mechanism and its use in the proof of Proposition~\ref{prop:NG
  cohomology} are given in Appendix~\ref{app:KO quartet mechanism} and
Appendix~\ref{sec:NG cohomology app} respectively. A limited and
elementary review of spectral sequences is also provided in Appendix
\ref{app:spectral sequences} as they are used to perform these
computations rigorously.

It is imperative to note that Proposition~\ref{prop:NG cohomology}
only holds for our specific choice of free field realisation of the
Nappi--Witten CFT given by \eqref{eq:Wakimoto embedding}. Different
choices of free field realisations lead to different null gauging
cohomology (an example is presented in Appendix~\ref{sec:NG cohomology alt app}). We elaborate on this point in Section~\ref{sec:FFR_comments}.

Before we proceed with the construction of the full string theory, we clarify the conformal symmetry of the matter content obtained from the null gauging procedure. Since we never used the conformal symmetry (i.e., $\Tsug$) of the Nappi--Witten CFT but rather just its operator algebra, we have not specified the Virasoro element of the null gauged theory which would enter a full string theory, specifically in the anti-holomorphic sector. 
It seems like one could use either $\bar{T}^{\mathrm{sug}}$ given by equation \eqref{eq:Tsug embedding} or simply $\bar{T}^{\beta\gamma}+\bar{T}^{\betil\gamtil}=\normalord{\betabar\barpar\gammabar}+\normalord{\betilbar\barpar\gamtilbar}$. 
We will proceed with using $\bar{T}^{\mathrm{sug}}$ in the next subsection, because this is the Virasoro element one would obtain from a WZW model, which is what we started with. 
Nonetheless, we note that both appear to be valid choices of Virasoro elements with the same central charges, and choosing either one is a choice of Virasoro representation in which the relative semi-infinite cohomology of the Virasoro algebra takes values. 
However, for reasons that will become clear, the choice between $\bar{T}^{\mathrm{sug}}$ and $\bar{T}^{\beta\gamma}+\bar{T}^{\betil\gamtil}$ does not affect the calculations in Section \ref{sec:full_theory}. In particular, we will observe that the space of states $V^{\betilbar\gamtilbar}_{\tilde{\sigma}}$ is left unchanged when we gauge the Virasoro symmetry (i.e., compute the BRST cohomology) of the full theory. That is, the change in the Virasoro algebra structure on $V^{\betilbar\gamtilbar}_{\tilde{\sigma}}$ due to the term $-\tfrac{1}{2}\barpar\betilbar$ is not reflected in the cohomology calculations.
We therefore summarise the result of the null gauging procedure in the following corollary.

\begin{corollary} \label{cor:gauged NW model hol anti-hol parts}
    The null chirally gauged Nappi--Witten model can be described by a
    worldsheet CFT whose holomorphic part is generated by free fields
    $(\beta,\gamma)$ and whose anti-holomorphic part is generated by
    free fields $(\bar{\beta},\bar{\gamma}, \betilbar,
    \gamtilbar)$. Hence, the holomorphic and anti-holomorphic parts of
    this CFT have central charges 2 and 4 respectively.
\end{corollary}

\subsection{The resulting string theory}\label{sec:full_theory}

Corollary~\ref{cor:gauged NW model hol anti-hol parts} reminds us of
the mismatch in the holomorphic and anti-holomorphic sectors of the
null chirally gauged Nappi--Witten model due to the gauging procedure
only enforcing a constraint in the former. Since the two sectors would contribute
different amounts of central charge to the worldsheet CFT of a full
bosonic string theory, any full string theory we build from the null
chirally gauged Nappi--Witten model would have different holomorphic
and anti-holomorphic field contents to account for this. 
The former requires the addition of a CFT with central charge
24 while the latter requires the addition of a CFT of central charge
22. A simple and physically sensible option is the addition of 22 free
bosons, each of which has a holomorphic and anti-holomorphic part, and
2 chiral bosons. Each free boson, labelled by an integer
$i\in\{1,\dots,22\}$, has holomorphic and anti-holomorphic parts given
by the free fields $\partial X^i(z)$ and
$\bar{\partial}\bar{X}^i(\bar{z})$ respectively.  The holomorphic and
anti-holomorphic matter sectors of our full string theory, along with
their Virasoro elements, are summarised in
Table~\ref{tab:matter_content}.
\begin{table}
    \centering
    \caption{Matter content of the null chirally gauged Nappi--Witten model}
    \label{tab:matter_content}
    \begin{tabular}{cc|cc}
         \multicolumn{2}{c|}{Holomorphic} & \multicolumn{2}{c}{Anti-holomorphic} \\ \hline 
        $X^1,\dots, X^{22}$  & $T^X= \tfrac{1}{2}\normalord{\partial X^i \partial X^i}$ & $\Xbar^1,\dots, \Xbar^{22}$ & $\bar{T}^X=\tfrac{1}{2}\partial \normalord{\Xbar^i \barpar \Xbar^i}$ \\ 
        $\beta, \gamma$ & $T^{\beta\gamma}= \normalord{\beta\partial\gamma}$ & $\betabar, \gammabar$ & $\bar{T}^{\beta\gamma}= \normalord{\betabar \barpar \gammabar}$\\
        $Y^1, Y^2$ & $T^Y =\tfrac{1}{2}\normalord{\partial Y^1 \partial Y^1 +\partial Y^2 \partial Y^2}$ & $\betilbar,\gamtilbar$ & $\bar{T}^{\betil\gamtil}_\text{mod} = \normalord{\betilbar \barpar \gamtilbar} - \frac{1}{2}\barpar\betilbar$\\
    \end{tabular}
    \vspace{1em}
    \caption*{The holomorphic and anti-holomorphic Virasoro elements are $\Tmatter=T^X+T^{\beta\gamma}+T^Y$ and $\Tbarmatter=\bar{T}^X + \bar{T}^\mathrm{sug}$, where we recall that $\bar{T}^\mathrm{sug}=\bar{T}^{\beta\gamma}+\bar{T}^{\betil\gamtil}_\text{mod}$.}
\end{table}

\subsubsection{The holomorphic sector}

The BRST current which enforces the constraint coming from the Virasoro gauge symmetry takes the well-known form of
 \begin{equation}
     \jVir = \normalord{c\Tmatter} + \tfrac{1}{2}\normalord{c\Tgh},
 \end{equation}
 where 
 \begin{equation}
     \Tmatter = T^X + T^{\beta\gamma} + T^Y, \quad \Tgh = -2\normalord{b\partial c}- \normalord{\partial b c},
\end{equation}
and $(b,c)$ are fermionic ghosts with weight $(2,-1)$ with respect to
$\Tgh$.  From Table \ref{tab:matter_content}, we can infer that we are
in a setting strikingly similar to the Gomis--Ooguri string
\cite{Gomis:2000bd}, in which the worldsheet matter content is given
by a $\beta\gamma$-system and 24 free Euclidean bosons. Since it is
known that one needs to compactify $\gamma \sim \gamma + 2\pi R$ to
obtain a non-trivial spectrum, we do the same here. This alters its
mode expansion by a winding mode\footnote{It does not alter the
  $\beta(z)\gamma(w)$ OPE, just as compactifying a free boson does not
  alter its OPE. The compactification therefore leaves the Operator product expansions
  \eqref{eq:gen NW OPE for P+P-}, \eqref{eq:gen NW OPE for J P+-} and
  \eqref{eq:gen NW OPE for JI} of the original fields in the
  Nappi--Witten CFT unchanged. See \cite[Chapter 8]{Polchinski:1998rq}
  for more details.}  
\begin{equation} \label{eq:gamma winding}
    \gamma(z) = iwR\ln(z) + \sum_n \gamma_n z^{-n}.
\end{equation}
Hence, the BRST operator written in out in terms
of modes is
\begin{equation} \label{eq:dVir split}
    \dVir = (\jVir)_0 = [\jVir,-]_1 = d_0+d_1,
\end{equation}
where
\begin{gather}
    d_0 = iwR\sum_{l \in\mathbb{Z}} c_{-l}\beta_l \label{eq:dVir_0} \\
    d_1 = -\sum_{\substack{l\in\mathbb{Z}\\k\in\mathbb{Z}}} k :c_{-l} \beta_{l-k} \gamma_{k}: + \sum_{l\in\mathbb{Z}} c_{-l} (L^X_l+L^Y_l)  + \sum_{\substack{l,k\in\mathbb{Z}\\k<l}} (k-l) : b_{k+l} c_{-l} c_{-k} : \label{eq:dVir_1}
\end{gather}
and $L^X_l$ and $L^Y_l$ are the modes of $T^X$ and $T^Y$
respectively. This splitting of $\dVir$ comes from assigning a
filtration degree to each of the modes from which, as alluded to
earlier, a spectral sequence is constructed to facilitate the
computation of cohomology. The BRST cohomology computations are
provided in full detail in Appendix~\ref{sec:GO string cohomology
  app}, but we state the splitting here to make the resemblance to the
Gomis--Ooguri string worldsheet more explicit (cf. \cite[Section
3.1]{Gomis:2000bd}).  Just as in the case of null gauging, the KO
mechanism greatly simplifies the computation, identical to what
happens in the Gomis--Ooguri string.

\begin{proposition} \label{prop:holo sector BRST}
  The BRST cohomology with respect to $\dVir$ is 
  \begin{equation}
    H^{\dotr}_{\mathrm{Vir}} \cong H^{\dotr}_{d_0} \cong \ket{\mathrm{vac}}_\sigma \otimes \mathcal{F}^X(k)\otimes\mathcal{F}^Y(\tilde{k})
  \end{equation}
  where $\ket{\mathrm{vac}}_\sigma = \ket{\sigma}_{\beta\gamma} \otimes
  \ket{1-\sigma}_{bc}$ is a choice of picture labelled by $m\in\ZZ$, and
  $\mathcal{F}^X(k)$ and $\mathcal{F}^Y(\tilde{k})$ are the Fock modules of
  the 22 free bosons $\{X^1,\dots X^{22}\}$ and 2 chiral bosons
  $\{Y^1,Y^2\}$ with momenta $k=(k_1,\dots, k_{22})$ and $\tilde{k}=(\tilde{k}_1,\tilde{k}_2)$
  respectively. 
\end{proposition}

\begin{proof}
  See Appendix~\ref{sec:GO string cohomology app}.
\end{proof}

\subsubsection{The anti-holomorphic sector}
Performing an identical compactification of $\gammabar$,
the form of the anti-holomorphic BRST current is identical to that of the holomorphic one
\begin{equation}
  \bar{j}_{\mathrm{Vir}} = \normalord{c\Tbarmatter} + \tfrac{1}{2}\normalord{c\Tbargh},
\end{equation}
with the only practical difference being the fact that 
\begin{equation}
  \Tbarmatter = \bar{T}^{X} + \bar{T}^{\mathrm{sug}}.
\end{equation}
Once again, we observe a simplification in BRST cohomology that can
be argued via spectral sequences and the KO mechanism, made possible
due to the compactification of $\bar{\gamma}$.

\begin{proposition} \label{prop:anti-holo sector BRST}
  The BRST cohomology with respect to $\dVir$ is 
  \begin{equation}
    H^{\dotr}_{\mathrm{Vir}} \cong H^{\dotr}_{d_0} \cong \ket{\mathrm{vac}}_\sigma \otimes \bar{\mathcal{F}}^X(k)\otimes \bar{V}^{\betil\gamtil}_{\tilde{\sigma}}
  \end{equation}
  where $\ket{\mathrm{vac}}_\sigma = \ket{\sigma}_{\betabar\gammabar} \otimes \ket{1-\sigma}_{\bbar\cbar}$ is a choice of picture labelled by $m\in\ZZ$, and $V^{\betil\gamtil}_{\tilde{\sigma}}$ contains a choice of $\betilbar\gamtilbar$-vacuum $\ket{\tilde{\sigma}}_{\betilbar\gamtilbar}$).
\end{proposition}

\begin{proof}
  See Appendix~\ref{sec:GO string cohomology app}.
\end{proof}

\subsection{Summary}

Overall, the resulting closed string theory can be interpreted as a
special case of the bosonic Gomis--Ooguri string. Recall that the
Gomis--Ooguri string spectrum was originally shown to be the Fock
module of 24 free bosons in euclidean space \cite{Gomis:2000bd}, up to
a choice of picture coming from a choice of $\beta\gamma$ vacuum. This
is shown explicitly in Appendix~\ref{sec:GO string cohomology app}. We
also show that for the purposes of computing the spectrum (i.e., BRST
cohomology), the choice of the $c=24$ matter sector that appears does
not matter. Starting with the low-energy limit of the bosonic string
and taking the appropriate limits forces the appearance of the 24 free
bosons. However, if we instead chose to start with the worldsheet
theory, the 24 free bosons simply go along for the ride and do not
alter the BRST cohomology calculations. From this perspective, any
$c=24$ CFT could have been chosen, but whether or not they admit the
sought after physical interpretation is a separate issue.

In that sense, coming back to the model we have constructed, it
resembles the Gomis--Ooguri string but with an altered $c=24$ matter
sector: 24 free bosons in the holomorphic sector,
and 22 free bosons and a $\betil\gamtil$-system in the
anti-holomorphic sector (note that the holomorphic sector seems to
admit the same galilean symmetry that the Gomis--Ooguri string
does).  The most peculiar feature of this model is the differing
holomorphic and anti-holomorphic sectors, but it is consistent with
the fact that this rendition of a Gomis--Ooguri-like string was
obtained from the null chiral gauging of a WZW model; the chirality of
the gauging procedure manifests itself in the full string theory, as
one would expect.

On that note, we address the free field realisation \eqref{eq:Wakimoto
  embedding} not being conformal due to the presence of the
$\tfrac{1}{2}\barpar\betilbar$ term. The key reason why the $c=24$
matter sectors in Propositions \ref{prop:holo sector BRST} and
\ref{prop:anti-holo sector BRST} are in no way reduced or changed due
to the cohomology calculations is precisely because the BRST
cohomology of both sectors is isomorphic to that of $d_0$. Hence, any
part of the embedding \eqref{eq:Wakimoto embedding} that enters the
$c=24$ part of the full string theory goes along for the ride and
remains unaltered. In our case, this includes the
$\betilbar\gamtilbar$-system in the anti-holomorphic sector. This is
exactly why the extra term in $\Tsug$ does not alter the Virasoro BRST
cohomology.

Finally, we also note that for the choice of free field realisation of the
Nappi--Witten CFT given in \eqref{eq:Wakimoto embedding}, the
computation of the null gauging cohomology followed by the computation
of Virasoro BRST cohomology is equivalent to the computation of the
BRST cohomology (or rather, semi-infinite cohomology) of
$\ghat_{\lambda=0} := \Vir \ltimes
\widehat{\mathfrak{u}(1)}$, where $\widehat{\mathfrak{u}(1)}$ is affine $\mathfrak{u}(1)$. 
Furthermore, the cohomology computations
are in keeping with \cite[Theorem 4.1]{Figueroa-OFarrill:2024wgs},
which means that the computation of the semi-infinite cohomology of
$\ghat_{\lambda=0}$ reduces to that of $\Vir$, by virtue of the
modules (i.e., free field realisations) being carefully chosen.  Both
statements are by no means true in general; they hold for our choice
of free field realisation as $(\betil,\gamtil, B,C)$ arrange
themselves into a ``quartet'' (see Appendix~\ref{app:KO quartet
  mechanism}). Nonetheless, it is an interesting observation that
non-relativistic string models can be realised as $\ghat_{\lambda=0}$
field theories.

\section{Discussion}
\label{sec:discussion}

In this paper, we presented an approach by which one can construct
strings propagating in galilean backgrounds using null, chirally
gauged WZW models on generalised Nappi--Witten Lie groups. In
particular, this approach does not rely on taking limits on the target
space or the worldsheet. We presented a large class of examples of WZW
models using which one could take the approach outlined in this
paper. For the special case of the Nappi--Witten model, the null
reduction of its target space is the three-dimensional galilean-AdS
spacetime, as will be shown in Appendix~\ref{sec:null-reds}.

The null chiral gauging of the Nappi--Witten model was then performed
using CFT techniques at the level of the worldsheet, as is standard
practice when studying WZW models. A convenient choice of free field
realisation was used to simplify the computations drastically.  The
resulting gauged WZW model was then used to a construct a ``full'' string
theory by adding to it some matter CFTs with the appropriate central
charges needed to gauge the worldsheet Virasoro symmetry. The final
string theory that we constructed can be interpreted as a closed
``Gomis--Ooguri-like'' string that displays a slight heterosis between
the holomorphic and anti-holomorphic sectors.

Given the choice of free field realisation \eqref{eq:Wakimoto
  embedding} for the Nappi--Witten CFT, our full string theory can be
interpreted as a $\Vir\ltimes\widehat{\mathfrak{u}(1)}\cong\ghat_{\lambda=0}$
field theory. Hence, our work demonstrates a possible string-theoretic
interpretation of $\ghat_{\lambda=0}$ field theories, which
complements the tensionless string interpretation of
$\bms\cong \ghat_{\lambda=-1}$ field theories. For
the sake of completeness, it is also worth mentioning the emergence of
$\ghat_{\lambda=-2}$ appearing as a symmetry algebra in the hybrid
formalism of the minimal tension ($k=1$) string in $\AdS_3\times S^3$
\cite{McStay:2023thk, McStay:2024dtk}. The fact that the one-parameter
family of Lie algebras $\ghat_\lambda$ appear as symmetry algebras in
different non-relativistic and tensionless string settings for
different values of the parameter $\lambda$ could be hinting at a
deeper, yet undiscovered connection between these theories, though as
it stands, their relationships do seem circumstantial.

\subsection{Comments on geometry}

The geometric motivation for this work involved making a choice of
lorentzian Lie group with a null isometry.  We then performed the
necessary manipulations and simplifications on the resulting conformal
field theory which we obtained from quantisation. However, it is worth
addressing that some of these simplifying steps led us
further away from the geometrical picture with which we started.  For
example, as we mentioned in Section~\ref{sec:Implement gauge
  constraint via BRST}, we chose to work with a different real form of
the Nappi--Witten Lie algebra, whose ad-invariant inner product has
split signature $(2,2)$ instead of $(3,1)$.  Hence, we expect any
string models we have constructed to admit pseudo-galilean symmetry
rather than a strictly galilean one. In fact, this seems to be
encapsulated by the differing holomorphic and anti-holomorphic sectors
of our full string theory. The presence of a non-trivial
$(\betil,\gamtil)$ spectrum in the anti-holomorphic sector should
contribute a lorentzian signature to the metric of any background
spacetime in which this string theory can be interpreted as
propagating. An intuitive way to infer this is from the embedding of
two lorentzian bosons $\d\phi^\pm$ into $(\betil,\gamtil)$ given by equations
\eqref{eq:field embeddings} and \eqref{eq:FFR mode embedding}. 
The lorentzian metric on $(\d\phi^+,\d\phi^-)$ is realised as the 
dual pairing on $(\betil,\gamtil)$, encapsulated in their OPEs.

It is worth making a few comments on this simple step of
complexification that we so routinely perform for the sake of
mathematical elegance and convenience.  Indeed, it is common practice
to complexify the Nappi--Witten algebra in a variety of different
settings \cite{Kiritsis:1993jk, Kiritsis:1994ij, Sfetsos:1993rh,
  Penafiel:2019czp} and bring it to the form given in
\eqref{eq:generalised_NW_complexified}. Even when studying the Lie
groups associated to the generalised Nappi--Witten algebra $\nw_{2\ell
  + 2}$ in \eqref{eq:generalised_NW_complexified}, this
confusion is prevalent. There is little distinction made between these
split signature $(\ell+1,\ell+1)$ Lie algebras and the centrally extended
euclidean algebra with Lorentzian signature $(1,2\ell+1)$, and it has
become common practice to go between them. At first glance, this
simply seems like going between two real forms of the same complex Lie
algebra. However, a lack of consistency in this practice is observed,
particularly since the desirable Lorentzian signature of one real form
of the generalised Nappi--Witten Lie algebra is used in conjunction with the diagonalisability of
$\ad J$ and triangular decomposition
(cf. \eqref{eq:generalised_NW_complexified}) of the other real form
\cite{Kiritsis:1993jk, Kiritsis:1994ij, Kehagias:1994iy,
  DAppollonio:2003zow, DAppollonio:2004ppq,
  DAppollonio:2007ldj}. Furthermore, it is the representation theory
of the complex Nappi--Witten Lie algebra that has been studied in
detail \cite{bao2011representations, Babichenko_2021,
  Chakraborty:2024hsk}, while the study of real representations of
real Lie algebras is much rarer in general. We hope that our work
serves as a reminder of the importance of this distinction in
physics. That being said, our work merely establishes a proof of
concept for the method of building non-relativistic strings as gauged
WZW models.

On the other hand, it would be worth noting the relevance of string
models on $(p> 1,q)$ backgrounds. Our work has inadvertently
demonstrated the emergence of string models that resemble the
Gomis--Ooguri string from gauged WZW models of strings propagating on
$(2,2)$ backgrounds. Other signature $(p,q)$ models have also appeared
in the literature earlier, such as the WZW model for the galilean
conformal algebra \cite{Chakraborty:2012qm} with signature $(3,3)$ and
WZW models for the unique indecomposable, nilpotent 5- and
6-dimensional metric Lie algebras \cite{Kehagias:1994ys}.

\subsection{Comments on free field realisations} \label{sec:FFR_comments}
As mentioned earlier, Proposition~\ref{prop:NG cohomology} only holds
when using our specific free field realisation of the Nappi--Witten
CFT. In fact, using the original free field realisaton written down in
\cite[Theorem 5.4]{bao2011representations} in terms of a
weight-$(1,0)$ $\beta\gamma$-system and two free Lorentzian bosons $\d
\phi^{\pm}$ yields a different null gauging cohomology. In particular,
the cohomology is ``smaller''. To understand this better, we recall
the free field realisation from
\cite[Theorem~5.4]{bao2011representations}:
\begin{equation}
\label{eq:Wakimoto embedding original}
    \begin{split}
        &P^+ = \beta\\
        &P^- = \partial \gamma +  \gamma\d\phi^- \\
        &I = \d\phi^-\\
        &J = \partial\phi^+ + \tfrac{1}{2}\d\phi^- - \normalord{\beta\gamma}. \\
    \end{split}
\end{equation} 
Comparing with the free field realisation \eqref{eq:Wakimoto embedding}, we have a sequence of embeddings
\begin{equation}\label{eq:field embeddings}
    \{P^\pm, I,J\} \hookrightarrow \{\d\phi^\pm, \beta,\gamma\} \hookrightarrow \{\betil,\gamtil,\beta,\gamma\},
\end{equation}
where the last embedding is given by the identification
\begin{equation} \label{eq:FFR mode embedding}
    \d\phi^+ = \d\gamtil \iff \alpha^+_n = -n \gamtil_n, \quad \d\phi^- = \betil \iff \alpha^-_n = \betil_n. 
\end{equation}
Note that $\gamma_0$ does not have a preimage under this embedding,
and that $\alpha^+_0 \mapsto 0$. The null gauging cohomology under the
free field realisation \eqref{eq:Wakimoto embedding original} is given
by the following proposition.

\begin{proposition} \label{prop:NG cohomology original FFR}
  The null gauging cohomology of $\{P^\pm,I,J\}$ under the embedding
  \eqref{eq:Wakimoto embedding original}, with respect to the
  differential given by the zero mode of
  \begin{equation}
    j_{\mathrm{NG}} = \tfrac{1}{2}\normalord{C\d \phi^-} - \normalord{C\beta\gamma} + \normalord{C\d\phi^+},
  \end{equation}
  is
  \begin{equation}
    H_{\mathrm{NG}}\cong\left(\CC\ket{0}_{BC} \oplus \CC C_0\ket{0}_{BC}\right)\otimes \CC \ket{0,0} \otimes V^{\beta\gamma}_{q;0}.
  \end{equation}
  Here, $\ket{p_+,p_-}$ denotes the Fock vacuum of the Lorentzian
  bosons $\d\phi^\pm$ with momentum $p_\mu = (p_+,p_-)$ and
  $V^{\beta\gamma}_{q;0}$ is the subspace of
  $V^{\beta\gamma}_\sigma$ which is the kernel of the zero mode of the
  current $\normalord{\beta\gamma}$.
\end{proposition}

In particular, there is no null gauging cohomology for $p_\mu =
(p_+,p_-) \neq (0,0)$, and only a subspace $V^{\beta\gamma}_{q;0}$
of the full space of states $V^{\beta\gamma}_\sigma$ remains, albeit there are
two copies of this subspace. In this sense, we say that the cohomology
is ``smaller'' than what we get from our free field realisaton
\eqref{eq:Wakimoto embedding}. For the interested reader, we leave the
proof of Proposition~\ref{prop:NG cohomology original FFR} in
Appendix~\ref{sec:NG cohomology alt app}. This calculation also
elucidates the fact that the $d_{NG}$-commuting diagonalisable
endomorphisms at our disposal change with the choice of free field
realisation, which hopefully gives some clarity regarding the
different end results of Propositions~\ref{prop:NG cohomology} and
\ref{prop:NG cohomology original FFR}.

\subsection{Future work}

We outline some possible future projects along the lines of this paper
or suggested by our results:
\begin{itemize}
    \item The construction of galilean string models via null
      reduction on WZW models on generalised Nappi--Witten Lie groups
      $\nw_{2\ell+2}$ for $\ell>1$, but keeping the lorentzian real
      form.
      
    \item A more detailed analysis of the resulting representation
      theory of the null gauging spectrum. Do states arrange
      themselves into representations of some symmetry group?
      
    \item The computation of null gauging cohomology without using
      free field realisations. A starting point would be to use Verma
      modules of the (affine) Nappi--Witten algebra, which were
      classified in \cite{bao2011representations, Babichenko_2021}.

    \item The determination of Lie groups admitting a bi-invariant
      stringy galilean/carrollian structure and perhaps a web of relations
      analogous to \eqref{eq:galcarbarg rels diagram}. This could
      provide a starting point to answer the question of whether there is a
      gauged WZW model description of Lie groups admitting a
      bi-invariant stringy galilean structure.

    \item One could also try to gauge null (but not necessarily
      central) subgroups of lorentzian Lie groups.  Such reductions do
      not correspond to galilean Lie groups, but only galilean
      homogeneous spaces.

    \item It would be interesting to extend this approach to construct
      galilean NSR strings.
      
\end{itemize}

The interpretation of other non-Lorentzian string theories as gauged WZW models is very much dependent on the choice of string theory. While it may be true that other string theories with galilean target spaces admit formulations as gauged WZW models on bargammanian geometries, this is not the case for string theories with carrollian target spaces. This is because a carrollian structure is a restriction (to a null hypersurface) of a bargmannian structure, while a galilean structure is a null reduction (i.e., quotient) of a bargmannian one. Unlike the latter, the former procedure cannot be realised as a gauging procedure of a WZW model. 
On the other hand, the techniques we used here to study the string worldsheet quantum mechanically can most certainly be applied to other non-Lorentzian strings. We have papers currently in preparation on the spectrum of the bosonic ambitwistor string (Carrollian worldsheet) and the Carrollian string (Carrollian worldsheet and target space).

\section*{Acknowledgements}

JMF would like to thank Emil Have and Niels Obers for interesting
discussions on the CFT approach to non-lorentzian string theories.
GSV would like to thank Gerben Oling and Ziqi Yan for brief yet useful
discussions on the Nappi--Witten algebra and non-relativistic string
geometries respectively.  GSV is supported by a Science and
Technologies Facilities Council studentship [grant number 2615874]. We would like to 
thank the anonymous referee for their question on the applicability of our techniques
to other non-Lorentzian strings.

\begin{appendices}

\section{Spectral sequences} \label{app:spectral sequences}

We provide a gentle introduction to spectral sequences and some of
their key features. It is based on \cite[Chapter 2]{MR1793722}
by McCleary. The discussion there makes use of $R$-modules for a
general commutative ring $R$, but for our purposes, we only work with
vector spaces which we will tacitly assume to be either real or
complex. It should be mentioned that this exposition is by no means a
thorough one; it is only intended to lay the groundwork for the
application of spectral sequences to the computation of string theory
spectra.

\subsection{Basic Setup}

\begin{definition}
  A \emph{differential bigraded complex} is a collection of bigraded
  vector spaces $\{E^{p,q}\}_{p,q\in\mathbb{Z}}$ with a linear
  differential $d \colon E^{\dotr,\dotr} \rightarrow E^{\dotr,\dotr}$
  of bidegree $(r, 1-r)$ for some $r\in\mathbb{Z}$ such that $d^2 :=
  d\circ d = 0$.
\end{definition}

\begin{definition}
  \label{def:spectral sequence}
  A (cohomological) \emph{spectral sequence} is a collection of
  differential bigraded complexes $\{E^{\dotr, \dotr}_r, d_r\}$ for $r
  = 0, 1, 2, \dots$, where for all $p,q,r$,
  \begin{equation}
    E^{p,q}_{r+1} \cong H^{p,q}(E_r^{\dotr,\dotr}, d_r) \defeq 
    \frac{\ker\big(d_r\colon E^{p,q}_r \longrightarrow E^{p+r,q+1-r}_r\big)}{\im\big(d_r\colon E^{p-r,q-1+r}_r\longrightarrow E^{p,q}_r\big)}.
  \end{equation}
\end{definition}

The collection of vector spaces $E^{\dotr,\dotr}_r$ is known as the
\emph{$r$\textsuperscript{th} page} of the spectral sequence. Figure
\ref{fig:spectral_seq_pages} illustrates the relationship between
the zeroth and first pages.

\begin{figure}[b]
    \centering
    \begin{tikzpicture}
    \foreach \x in {0,...,6}
        \foreach \y in {0,...,6}
        {\fill (\x,\y) circle (3pt);
        \draw[-to](\x, \y)--(\x+0.9,\y);}
    \node[below] at (0,0) {$E^{0,0}_0$};
    \node[below] at (1,0) {$E^{0,1}_0$};
    \draw[-to, color=BlueViolet, line width=0.4mm](2,2)--(2.9,2);
    \draw[-to, color=BlueViolet, line width=0.4mm](3,2)--(3.9,2);
    \fill[color=BlueViolet] (3,2) circle (3.3pt);
    \node[below, text=BlueViolet] at (3,2) {$E^{2,3}_0$};
    \foreach \x in {9,...,15}
        \foreach \y in {0,...,6}
        {\fill (\x,\y) circle (3pt);
        \draw[-to](\x, \y)--(\x,\y+0.9);}
    \fill[color=BlueViolet] (12,2) circle (3.3pt);
    \node[left, text=BlueViolet] at (12,2) {$E^{2,3}_1$};
    \end{tikzpicture}
    \vspace{5pt}
    \caption{The zeroth page (left) and the first page (right) of a
      spectral sequence. The grading $p$ increases along the $y$-axis
      while $q$ increases along the $x$-axis and the arrows depict the
      action of $d_r$. The cohomology at $E^{2,3}_0$ determines
      $E^{2,3}_1$ as highlighted in blue. However, note that it does
      \textit{not} determine $d_1$.}
    \label{fig:spectral_seq_pages}
\end{figure}

We would like identify the end goal of computations with spectral
sequences. The most intuitive thing to think about would be to keep
turning the pages of a spectral sequence until one reaches a situation
where the cohomology does not change.  This can be made more precise.
We suppress the $(p,q)$ indices for clarity and start with the second page of
the spectral sequence for concreteness.  Define $Z_2
\defeq \ker d_2$ and $B_2 \defeq \im d_2$. Then
\begin{equation*}
  d_2 \circ d_2 = 0 \implies B_2 \subset Z_2 \subset E_2.
\end{equation*}
Now define $\Bar{Z}_3 \defeq \ker d_3$ and $\Bar{B}_3 \defeq \im
d_3$. Since these are subspaces of $E_3 = \bigslant{Z_2}{B_2}$, there
exist $B_3 \subset Z_3 \subset E_3$ such that $\Bar{Z}_3 =
\bigslant{Z_3}{B_2}$ and $\Bar{B}_3 =
\bigslant{B_3}{B_2}$. Consequently, $E_4 =
\bigslant{\Bar{Z}_3}{\Bar{B}_3} = \bigslant{Z_3}{B_3}$ and we have the
tower of inclusions
\begin{equation*}
  B_2 \subset B_3 \subset Z_3 \subset Z_2.
\end{equation*}
Iterating, the spectral sequence can be presented as an infinite tower
of inclusions
\begin{equation*}
  B_2 \subset B_3 \subset \cdots \subset B_n \subset B_{n+1} \subset \cdots \subset
  Z_{n+1} \subset Z_n \subset \cdots \subset Z_3 \subset Z_2
\end{equation*}
with $E_{n+1}=\bigslant{Z_n}{B_n}$, and differentials
\begin{equation*}
  d_{n+1} \colon \bigslant{Z_n}{B_n} \rightarrow \bigslant{Z_n}{B_n} \quad \quad \ker d_{n+1} = \bigslant{Z_{n+1}}{B_n} \quad \quad \im d_{n+1} = \bigslant{B_{n+1}}{B_n}.
\end{equation*}
This gives rise to a short exact sequence for each $n$
\begin{equation}
\label{eq:spectralseq_SES}
  0 \xrightarrow{\makebox[2em]{}} \bigslant{Z_{n+1}}{B_{n}} \xrightarrow{\makebox[1.8em]{}} \bigslant{Z_n}{B_n} \xrightarrow{\makebox[2em]{$d_{n+1}$}} \bigslant{B_{n+1}}{B_n} \xrightarrow{\makebox[2em]{}} 0.  
\end{equation}
An element in $E_2$ is said to \emph{survive to the
  $r$\textsuperscript{th} page} if it lies in $Z_r$. Likewise, an
element in $E_2$ is said to \emph{bound at the
  $r$\textsuperscript{th} page} if it lies in $B_r$. We may then
define 
\begin{equation}
    Z_{\infty} \defeq \bigcap_n Z_n \quad \quad B_\infty \defeq \bigcup_n B_n
\end{equation}
as the subspaces of $E_2$ consisting of elements that \emph{survive
  forever} and \emph{eventually bound}, respectively. Looking back at
the infinite tower of inclusions, we indeed have $B_\infty \subset
Z_\infty$ and thus, we may sensibly define:
\begin{equation}
    E_\infty \defeq \bigslant{Z_\infty}{B_\infty}.
\end{equation}
This is the bigraded vector space that one obtains from the calculation of
the infinite sequence of successive cohomologies. In some cases, this
calculation truncates at some finite stage.

\begin{definition} \label{def:spectral_sequence_collapse}
  A spectral sequence is said to \emph{collapse at the
    $N$\textsuperscript{th} page} if $d_{r\geq N} = 0$.
\end{definition}

From \eqref{eq:spectralseq_SES}, $d_{r=N} = 0 \implies Z_N = Z_{N-1}$,
since the kernel of the zero map $d_N$ is simply
$\bigslant{Z_{N-1}}{B_{N-1}}$ which should equal the image
$\bigslant{Z_N}{B_{N-1}}$ of the inclusion map. Likewise,
\begin{equation*}
  \ker\left(\bigslant{B_N}{B_{N-1}} \longrightarrow 0 \right) = \bigslant{B_N}{B_{N-1}} = \im d_N = 0 \iff B_N = B_{N-1}.
\end{equation*}
The tower of inclusions then becomes 
\begin{equation*}
  B_2 \subset B_3 \subset \dots \subset B_{N-1} = B_{N} = \dots = B_\infty \subset Z_\infty = \dots = Z_N = Z_{N-1} \subset \dots \subset Z_3 \subset Z_2,
\end{equation*}
thereby reducing the computation to a finite one as promised. It is
this collapse that we will exploit when using spectral sequences to
compute string theory spectra.

\subsection{Spectral sequences from filtrations}

One way a spectral sequence can arise is through a filtration. Suppose
$\{A^{\dotr} = \bigoplus_{n\in\mathbb{Z}} A^n, d\}$ is a differential
graded vector space with $d\colon A^n \rightarrow A^{n+1}$. Then we
apply a decreasing filtration (commonly the case for cohomological
spectral sequences)
\begin{equation}
    F : \dots \subseteq F^{p+1}A \subseteq F^{p}A \subseteq F^{p-1}A \subseteq \cdots
\end{equation}
such that the differential respects the filtration (so $d\colon F^p A
\rightarrow F^p A$). We may also define the filtration at each graded
level via $F^p A^n \defeq F^p A \cap A^n$. Then $d\colon F^p A^n
\rightarrow F^{p}A^{n+1}$.  Every filtered vector space $F^{\dotr} A$ has an
associated graded vector space $\Gr_{\dotr} A$, where
\begin{equation}
  \label{eq:asssociated-graded}
  \Gr_p A := \bigslant{F^p A}{F^{p+1} A}.
\end{equation}

\begin{figure}[t]
    \centering
    \begin{tikzpicture}
    \draw[very thick] (-0.5,0.3) rectangle (2,3.6);
    \draw[very thick] (6,0.3) rectangle (8.5,3.6);
    \node at (0.75,3.95) {$A^{n}$};
    \node at (7.25,3.95) {$A^{n+1}$};
    \draw[very thick, dashed] (-0.5,2.7)--(2,2.7);
    \draw[very thick, dashed] (6,2.7)--(8.5,2.7);
    \draw[-to, blue] (1.7,2.3)--(6.3,2.3);
    \node at (0.75, 2.3) {$F^p A^n$};
    \node at (7.25, 2.3) {$F^p A^{n+1}$};
    \draw[very thick, dashed] (-0.5,1.9)--(2,1.9);
    \draw[very thick, dashed] (6,1.9)--(8.5,1.9);
    \draw[very thick, dashed] (-0.5,1.3)--(2,1.3);
    \draw[very thick, dashed] (6,1.3)--(8.5,1.3);
    \node at (0.75, 1.6) {$F^{p+1}A^{n}$};
    \node at (7.25, 1.6) {$F^{p+1}A^{n+1}$};
    \draw[-to, blue] (1.7,1.6)--(6.3,1.6);
    \draw[-to] (1.7,2.3)--(6.3,1.7);
    \node at (0.75,1) {$\vdots$};
    \node at (7.25,1) {$\vdots$};
    \draw[-to] (1.7,2.3)--(6.3,1.1);
    \draw[-to] (1.7,1.6)--(6.3,1);
    \node at (0.75,3.25) {$\vdots$};
    \node at (7.25,3.25) {$\vdots$};
    \end{tikzpicture}
    \vspace{5pt}
    \caption{The arrows depict the possible actions of $d$. The filtration degree of any element in $A^n$ is never raised by $d$. The blue arrows depict $d_0$.}
    \label{fig:d respect filtration}
\end{figure}

\begin{definition} \label{def:assoc_bigraded_module}
  For any decreasing filtration $F$ on a differential graded complex $\{A^{\dotr}, d\}$, the \emph{associated bigraded complex} is given by 
  \begin{equation}
    E^{p,q}_0(A^{\dotr}, F) \defeq \Gr_p A^{p+q}
  \end{equation}
  relative to the differential $d_0$ induced by $d$.
\end{definition}
For our purposes, the bigraded complex associated to a filtration $F$
will be how the zeroth page of a spectral sequence emerges (hence the
compatible notation), so let us proceed under the assumption that such
a spectral sequence exists, with $d_0$ the part of $d$ that leaves the
filtration degree unchanged.

Now consider Figure~\ref{fig:filtration and differential}.
\begin{figure}
    \centering
    \begin{tikzcd}
        &                    & \vdots                       & \vdots                     & \vdots                       & \\
        &\dots\arrow[r,"d"] & F^{p-1}A^{n-1}\arrow[r, "d"]\arrow[u, hookrightarrow] & F^{p-1}A^{n}\arrow[r, "d"]\arrow[u, hookrightarrow] & F^{p-1}A^{n+1}\arrow[r, "d"]\arrow[u, hookrightarrow] &\dots \\
        &\dots\arrow[r,"d"] & F^{p}A^{n-1}\arrow[r, "d"]\arrow[u, hookrightarrow] & F^{p}A^{n}\arrow[r, "d"]\arrow[u, hookrightarrow] & F^{p}A^{n+1}\arrow[r, "d"]\arrow[u, hookrightarrow] &\dots \\
        &\dots\arrow[r,"d"] & F^{p+1}A^{n-1}\arrow[r, "d"]\arrow[u, hookrightarrow] & F^{p+1}A^{n}\arrow[r, "d"]\arrow[u, hookrightarrow] & F^{p+1}A^{n+1}\arrow[r, "d"]\arrow[u, hookrightarrow] &\dots \\
        &                    & \vdots\arrow[u, hookrightarrow]                     & \vdots\arrow[u, hookrightarrow]                     & \vdots\arrow[u, hookrightarrow]                       &
    \end{tikzcd}
    \caption{The structure of the filtered differential graded complex with the action of $d$.}
    \label{fig:filtration and differential}
\end{figure}
There are two natural operations that we can perform: taking quotients
along the inclusions and computing cohomologies along $d$. If we do
the latter first, we observe that there is an induced filtration
(which we will also denote by $F$)  on  the cohomology $H^{\dotr}(A,d)$:
\begin{equation}
\label{eq:induced_filtration_cohomology}
    F^p H^n (A,d) \defeq H^n (F^p A, d).
\end{equation}
This amounts to taking the cohomology at each filtration degree. As a
result $\{H^{\dotr}, d, F\}$ is also a filtered differential complex,
from which we may construct the associated graded complex  $E^{p,q}_0
\big(H^{\dotr}(A,d),F)\big)$. If instead we take quotients first, we
obtain the complex $\{E^{\dotr, \dotr}_0 (A^{\dotr},F),d_0\}$:
\begin{equation}
  \label{eq:assoc_grad_mod_cohom_diagram}
  \begin{tikzcd}
    &\dots\arrow[r,"d_0"] & \Gr_pA^{n-1} \arrow[r, "d_0"] & \Gr_p A^n \arrow[r, "d_0"]  & \Gr_p A^{n+1} \arrow[r, "d_0"] &\dots
  \end{tikzcd}
\end{equation}
The differential is now $d_0$ instead of $d$, because $d_0$ is the
only part of $d$ that is non-vanishing on the quotients since the rest
of $d$ would raise the filtration degree. Letting $n=p+q$ for clarity,
we obtain
\begin{equation}
    H^{p+q} \left( \Gr_p A, d_0\right) = \frac{\ker \left(d\colon
        \Gr_p A^{p+q} \longrightarrow \Gr_p A^{p+q+1}\right)}{\im \left(d\colon \Gr_p A^{p+q-1} \longrightarrow \Gr_p A^{p+q}\right)}.
\end{equation}
The RHS is precisely $H^{p,q} \big(E^{\dotr,\dotr}_0 (A,F), d_0\big)
\cong E^{p,q}_1$, which would be the first page of our spectral
sequence. Hence, we have obtained two different objects, $E^{p,q}_0
\big(H^{\dotr}(A,d),F)\big)$ and $H^{p,q} \big(E^{\dotr,\dotr}_0
(A,F), d_0\big)$, and the following equivalent statements by taking
quotients and cohomologies in a different order:
\begin{enumerate}
    \item the associated bigraded complex $E^{p,q}_0 (A^{\dotr}, F)$ is the zeroth page of a spectral sequence;
    \item and the first page of a spectral sequence $E^{p,q}_1$ is isomorphic to $H^{p+q} \left(\Gr_p A, d_0\right)$.
\end{enumerate}

Recall that we had constructed $E_\infty$ by laying out the spectral
sequence as an infinite tower of inclusions.  We now introduce the
notion of convergence.

\begin{definition} \label{def:convergence}
Let $\mathcal{A}^{\dotr}$ be a graded vector space. A spectral sequence \emph{converges} to $\mathcal{A}^{\dotr}$ if there exists a filtration $\mathcal{F}$ on $\mathcal{A}^{\dotr}$ such that
\begin{equation}
  E^{p,q}_\infty \cong E^{p,q}_0 (\mathcal{A}^{\dotr}, \mathcal{F}).
\end{equation}
\end{definition}

The interplay between $E^{p,q}_\infty$, $E^{p,q}_0
\big(H^{\dotr}(A,d),F)\big)$ and $H^{p,q} \big(E^{\dotr,\dotr}_0
(A,F), d_0\big)$ is encapsulated in the following theorem.

\begin{theorem}
\label{thm:Spectral Sequence Converges to Cohomology}
Each filtered differential graded complex $\{A^{\dotr}, d, F\}$
determines a spectral sequence $\{E^{\dotr, \dotr}_r,
d_r\}_{r\in\mathbb{N}^+}$ with
\begin{equation}
  E^{p,q}_1 \cong H^{p+q} \left(\Gr_p A, d_0\right) = H^{p,q} \big(E^{\dotr,\dotr}_0 (A,F), d_0\big).
\end{equation}
Suppose also that $F$ is a \emph{bounded filtration}, meaning $\exists s(n)$, $t(n)\in \mathbb{Z}$ for every $n$ such that 
\begin{equation}
  0 = F^{s(n)}A^n \subseteq F^{s(n)-1}A^n \subseteq \dots \subseteq F^{t(n)+1}A^n \subset F^{t(n)} A^n = A^n.
\end{equation}
Then the spectral sequences converges to $H^{\dotr}(A,d)$:
\begin{equation}
  E^{p,q}_\infty \cong E^{p,q}_0 \big(H^{\dotr}(A,d),F)\big) = \Gr_p H^{p+q}(A,d).
\end{equation}
\end{theorem}

This is the key theorem from spectral sequences which was utilised in
\cite{MR865483} to prove the vanishing theorem in semi-infinite
cohomology with values in a hermitian module.  Essentially,
$H^{\dotr}(A,d)$ plays the role of $\mathcal{A}^{\dotr}$ in
Definition~\ref{def:convergence} to formulate the last statement of
the theorem.  A detailed proof is given in McCleary's book \cite[Chapter
2.2]{MR1793722}. However, examination of the proof reveals that
it is often too strong to demand that the filtration be bounded.

\begin{definition}
  Let $\{A^{\dotr} = \oplus_{n\in\mathbb{Z}} A^n, d\}$ be a
  differential graded complex and let $F$ be a stable filtration
  (preserved by $d$).  We say that $F$ is \emph{exhaustive} if
  $\bigcup_p F^P A = A^{\dotr}$ and \emph{weakly convergent} if
  $\bigcap_p F^p A = 0$.
\end{definition}

Theorem~\ref{thm:Spectral Sequence Converges to Cohomology} actually
holds for an exhaustive and weakly convergent filtration.  We make use
of this version of the theorem when computing the BRST cohomology of
non-relativistic strings.

\section{The Kugo--Ojima quartet mechanism} \label{app:KO quartet mechanism}
\subsection{Elementary version} \label{sec:KO_basic}

The simplest way to describe the KO mechanism is as follows. Consider
two sets of creation and annihilation operators, $(a^\dagger,a)$ and
$(b^\dagger, b)$, the former bosonic and the latter fermionic, obeying
the following non-trivial (anti-)commutation relations:
\begin{equation}
   [a,a^\dagger] = 1 \iff a a^\dagger = 1 + a^\dagger a  \quad
    [b, b^\dagger] = 1 \iff bb^\dagger = 1 - b^\dagger b  
\end{equation}
These act on $V$, the vector space of states generated by the creation
operators\footnote{To clarify, the dagger here does not denote
  Hermitian conjugation. It is simply notation to denote a creation
  operator.}  $a^\dagger$ and $b^\dagger$ acting on a vacuum state
$\ket{0}$ which obeys
\begin{equation}
    a\ket{0} = b\ket{0} = 0.
\end{equation}
Let $Q = a^\dagger b$.   It follows that $Q^2 = 0$ and we would like
compute its cohomology. We may try to exhibit a smaller
quasi-isomorphic complex by somehow ruling out states that are
manifestly $Q$-exact. One way to find manifestly $Q$-exact states is
to construct a diagonalisable endomorphism $\Phi\in\End V$ such that
\begin{equation}
    \Phi = Q \varphi + \varphi Q,
\end{equation}
for some $\varphi\in\End V$.  By construction, $\Phi$ commutes with $Q$
(i.e., $[Q,\Phi]=0$), so $Q$ preserves $\Phi$-eigenspaces.  This means
that the complex breaks up into a direct sum of subcomplexes
corresponding to the different eigenvalues of $\Phi$.  We claim that
the cohomology resides in $\ker\Phi$.  Indeed, consider any
$Q$-cocycle $\ket{\psi}$ which is also an eigenstate of $\Phi$ with
nonzero eigenvalue $\lambda \neq 0$.  Then,
\begin{equation} \label{eq:non-zero-eigenstate-coboundary}
  \Phi\ket{\psi} = \lambda\ket{\psi} = (Q\varphi + \varphi Q)
  \ket{\psi} = Q\varphi\ket{\psi},
\end{equation}
so that
\begin{equation}
  \ket{\psi} = Q\left(\tfrac{1}{\lambda}\varphi \ket{\psi}\right)
\end{equation}
is actually a $Q$-coboundary.  This means that the cohomology of $Q$
can be calculated from the subcomplex $\ker \Phi$.

Indeed, if a suitable choice of $\Phi\in\End V$ is available, this
technique could help us simplify the computation of $Q$-cohomology
drastically. With this in mind, we define $K \defeq b^\dagger a$,
using which we construct the operator
\begin{equation}
\begin{split}
   \mathbf{N} \defeq & QK + KQ  \\
        =& a^\dagger b b^\dagger a + b^\dagger a a^\dagger b\\
        =& (1 - b^\dagger b) a^\dagger a + b^\dagger b (1 + a^\dagger a) \\
        =& a^\dagger a + b^\dagger b.
\end{split}
\end{equation}
The operator $\mathbf{N}$ is called the \emph{number operator}, since
it counts the number of creation operators acting on $\ket{0}$:
\begin{equation} \label{eq:Num_op_basic}
\begin{gathered}
    \mathbf{N} (a^\dagger \ket{0}) = a^\dagger (1+a^\dagger a)\ket{0} = a^\dagger \ket{0}\\
    \mathbf{N} (b^\dagger \ket{0}) = b^\dagger (1-b^\dagger b)\ket{0} = b^\dagger \ket{0}\\
    \mathbf{N}(b^\dagger)^l(a^\dagger)^k\ket{0}= (k+l)(b^\dagger)^l(a^\dagger)^k\ket{0}.
\end{gathered}
\end{equation}
Let $\ket{\phi_{k+l}} = \defeq(b^\dagger)^l(a^\dagger)^k\ket{0}$,
where $k\in\NN$ and $l\in\{0,1\}$. The set
$\{\ket{\phi_n}\}_{n\in\NN}$ not only form a basis for $V$, but are
also eigenstates of $\mathbf{N}$ with $\ket{\phi_n}$ having eigenvalue $n$, as shown in
\eqref{eq:Num_op_basic}. This shows that $\mathbf{N}$ is
diagonalisable over $V.$  Hence, following the same logic presented in
\eqref{eq:non-zero-eigenstate-coboundary}, we can compute the
$Q$-cohomology from the subcomplex $\ker \mathbf{N}$.  But
$\ker\mathbf{N}$ is one-dimensional and spanned by $\ket{0}$, which is
a $Q$-cocycle, so that
\begin{equation}
    H_Q \cong \mathbb{C}\ket{0}.
\end{equation}
This reduction of cohomology to a one-dimensional space spanned by the
vacuum state is known as the Kugo--Ojima quartet mechanism (which we
abbreviate to KO mechanism), and the operators
$(a,a^\dagger,b,b^\dagger)$ are said to form a \emph{quartet}.

\subsection{Full version} \label{sec:KO_full}

In our examples, the object of interest is the \emph{Koszul CFT}.

\begin{definition}
\label{def:Koszul CFT}
A \emph{Koszul CFT} consists of a bosonic $\beta\gamma$-system and a
fermionic $bc$-system of weights $(\lambda,1-\lambda)$ and
$(\mu,1-\mu)$, respectively, together with a differential
$d_{\mathrm{KO}} \defeq \normalord{c\beta}_0$, whose cohomology
$H^{\dotr}_{\mathrm{KO}}$ is called the \emph{chiral ring} of the Koszul
CFT.
\end{definition}

The KO mechanism is responsible for the one-dimensionality of the
chiral ring of the Koszul CFT.

\begin{lemma}
    \label{lem:Koszul chiral ring}
    The chiral ring of a Koszul CFT is one-dimensional. That is,
    \begin{equation}
        \label{eq:Koszul CFT cohomology}
        H^{n}_{\mathrm{KO}}=\begin{cases}
             \CC\ket{\mathrm{vac}}_\sigma, \quad n=0 \\
             0, \quad \text{otherwise.}
        \end{cases}
    \end{equation}
    Here, $\ket{\mathrm{vac}}_\sigma= \ket{-\sigma}_{\beta\gamma} \otimes \ket{\sigma+\mu-\lambda}_{bc}$ denotes a choice of picture, labelled by an integer $\sigma\in\ZZ$.
\end{lemma}

\begin{proof}
    Recall that the space of states of the Koszul CFT is spanned by
    monomials constructed by the appropriate modes acting on a choice
    of $\beta\gamma$ vacuum $\ket{\sigma}_{\beta\gamma}$. Although the $bc$
    vacua are equivalent, since we can go between them by acting $bc$
    modes, there is a convenient choice corresponding to each
    $\ket{\sigma}_{\beta\gamma}$, which simply comes from imposing a
    ``vacuum-matching'' condition. That is, if $\beta_n
    \ket{\sigma}_{\beta\gamma} = 0$ for some $n\in\mathbb{Z}$, then $b_n
    \ket{\rho}_{bc} = 0$ too, and likewise for $\gamma_n$ and $c_n$
    acting on the respective vacua. From \eqref{eq:bcbetagamma vacua},
    this enforces $\rho=-\sigma+\mu-\lambda$. Hence, we choose our vacuum
    state to be
    $\ket{\mathrm{vac}}_\sigma=\ket{\sigma}_{\beta\gamma}\otimes\ket{-\sigma+\mu-\lambda}_{bc}$,
    where $\sigma\in\ZZ$ is a choice.  This choice allows for the most
    natural formulation of the reduction in cohomology via the KO
    mechanism.

    Next, we want to write $d_{\mathrm{KO}}$ in a similar form to the
    differential $Q$ in the previous subsection. This would first
    require placing all annihilators (with respect to
    $\ket{\mathrm{vac}}$) to the right:
    \begin{equation} \label{eq:dKO split}
        d_{\mathrm{KO}} = \normalord{c\beta}_0 = \sum_{l\in\ZZ} c_{-l} \beta_l 
        = \sum_{l\geq -\sigma+1-\lambda}c_{-l}\beta_l + \sum_{l\geq \lambda+\sigma}\beta_{-l} c_{l}
    \end{equation}
    Notice how there are two terms in \eqref{eq:dKO split}, with each
    term schematically of the form taken by $Q$ in the previous
    subsection (i.e., consisting of a creation operator and an
    annihilation operator with opposite parity to each other). This
    because there are four pairs of creation and annihilation
    operators, instead of two in the previous subsection, with each
    pair labelled by an integer. Two of them are bosonic,
    corresponding to the modes $\beta_n$ and $\gamma_n$, while the
    other two are fermionic, corresponding to the modes $b_n$ and
    $c_n$. An explicit assignation of ``creation'' or ``annihilation''
    to each mode $\beta_n$, $\gamma_n$, $b_n$, $c_n$ can be inferred
    from the form of $d_{\mathrm{KO}}$ in \eqref{eq:dKO split}, summarised
    below:
    \begin{equation}
    \label{eq:KO mechanism cre-ann ops}
    \begin{alignedat}{4}
            &a^\dagger(\beta_n) = \beta_{-n} \quad &&n\geq \lambda+\sigma \quad\quad &&a^\dagger(\gamma_n)= -\gamma_{-n} \quad &&n\geq -\sigma+1-\lambda\\
            &a(\beta_n) = \gamma_n \quad &&n\geq \lambda+\sigma \quad\quad &&a(\gamma_n)= \beta_n\quad &&n\geq -\sigma+1-\lambda\\
            &a^\dagger(b_n) = b_{-n} \quad &&n\geq \lambda+\sigma \quad\quad &&a^\dagger(c_n)= c_{-n} \quad &&n\geq -\sigma+1-\lambda\\
            &a(b_n) = c_n \quad &&n\geq \lambda+\sigma \quad\quad &&a(c_n)= b_n \quad &&n\geq -\sigma+1-\lambda.
    \end{alignedat}    
    \end{equation}
    These operators are indeed compatible with the mode algebra of the $\beta\gamma$ and $bc$-systems, which can be expressed as
    \begin{equation}
        \begin{alignedat}{2}
            &[a(\phi_k), a^\dagger(\phi_l)] &&= \delta_{kl} \quad \quad \phi = \beta,\, \gamma \\
            &[a(\phi_k), a^\dagger(\phi_l)]_+ &&= \delta_{kl} \quad \quad \phi = b,\, c 
        \end{alignedat}
    \end{equation}
    Then using \eqref{eq:KO mechanism cre-ann ops},
    \begin{equation}
        d_{\mathrm{KO}} = \sum_{n\geq -\sigma+1-\lambda} a^\dagger(c_n)a(\gamma_n) + \sum_{n\geq \lambda+\sigma} a^\dagger(\beta_n)a(b_n).
    \end{equation}
    In this form, the resemblance of $d_{\mathrm{KO}}$ to $Q$ is
    explicit. This lets us construct an operator $K$ from the
    expression above for $d_{\mathrm{KO}}$ by mirroring the corresponding
    expression in the previous subsection given in terms of creation
    and annihilation operators
    \begin{equation}
      K = \sum_{n\geq \sigma+1-\lambda} a^\dagger(\gamma_n)a(c_n) + \sum_{n\geq \lambda-\sigma} a^\dagger(b_n)a(\beta_n).
    \end{equation}
    Finally, we may construct the analogous number operator as the anti-commutator of $d_{\mathrm{KO}}$ and $K$
    \begin{equation}
      \label{eq:OG GO N_Tot}
      \begin{split}
        N_{\mathrm{tot}} &= d_{\mathrm{KO}} K + K d_{\mathrm{KO}}\\
        &= \sum_n a^\dagger(\beta_n) a(\beta_n) + \sum_n a^\dagger(\gamma_n) a(\gamma_n) + \sum_n a^\dagger (b_n) a(b_n) + \sum_n a^\dagger (c_n) a(c_n). 
      \end{split}
    \end{equation}
    Indeed, we see that $N_{\mathrm{tot}}$ also resembles $\mathbf{N}$
    from the previous subsection, but let us quickly show that
    $N_{\mathrm{tot}}$ is diagonalisable on $\mathcal{V}\defeq
    V^{\beta\gamma}_\sigma \otimes V^{bc}$.  Recall that $\mathcal{V}$ is
    spanned by monomials of the form
    \begin{equation}
      \ket{\omega} = b_{-n_1}\dots b_{-n_\mathcal{B}} c_{-m_1}\dots c_{-m_\mathcal{C}}\ket{\sigma+\mu-\lambda}_{bc}
      \otimes
      \beta_{-r_1}\dots \beta_{-r_\mathfrak{B}} \gamma_{-s_1}\dots\gamma_{-s_\mathfrak{C}}\ket{-\sigma}_{\beta\gamma},   
    \end{equation}
    where 
    \begin{gather*}
      n_1 > \dots > n_{\mathcal{B}}  \geq \lambda+\sigma,\ \ m_1 > \dots > m_{\mathcal{C}} \geq -\sigma+1-\lambda\\
      r_1 \geq \dots \geq r_{\mathfrak{B}}\geq \lambda+\sigma,\ \  s_1 \geq \dots \geq s_{\mathfrak{C}}\geq -\sigma+1-\lambda.
    \end{gather*}
    But each monomial is an eigenstate of $N_{\mathrm{tot}}$
    \begin{equation}
      N_{\mathrm{tot}}\ket{\omega}  = (\mathfrak{B}+\mathfrak{C}+\mathcal{B}+\mathcal{C})\ket{\omega}.
    \end{equation}
    Hence, $N_{\mathrm{tot}}$ is diagonalisable and $d_{\mathrm{KO}}$-exact.
    Following identical arguments to those in the previous subsection,
    we conclude that the cohomology can be computed from the
    subcomplex $\ker N_{\mathrm{tot}}$.  But $\ker N_{\mathrm{tot}}$ is
    one-dimensional and spanned by $\ket{\mathrm{vac}}$, which is a
    $d_{\mathrm{KO}}$-cocycle. This completes the proof.
\end{proof}

\section{Cohomology computations}
\label{sec:cohom-comp}

Using the techniques and results presented in
Appendices~\ref{app:spectral sequences} and \ref{app:KO quartet
  mechanism}, we present the computation of null gauging cohomology
(i.e., a proof of Proposition~\ref{prop:NG cohomology}) and the BRST
cohomology of the Gomis--Ooguri string. The latter proves
Propositions~\ref{prop:holo sector BRST} and \ref{prop:anti-holo
  sector BRST} since, as we argued earlier, the matter content in the
theories from which these arise look like special cases of the
matter content of the closed Gomis--Ooguri string.

\subsection{Null gauging cohomology} \label{sec:NG cohomology app}

Recall that the space of states of the Nappi--Witten CFT is $V\defeq
V^{\beta\gamma}_\sigma\otimes V^{\betil\gamtil}_{\tilde{\sigma}}\otimes V^{BC}$,
spanned by monomials for the form \eqref{eq:monomials of full space
  Wakimoto}.  We filter $V$ by assigning a filtration degree to each
mode as follows:
    \begin{center}
        \begin{tabular}{ c|c c c c c c} 
             & $\betil_n$ & $\gamtil_n$ & $B_n$ & $C_n$ & $\beta_n$ & $\gamma_n$ \\ \hline
            $\fdeg$ & -1 & 1 & -1 & 1 & 0 & 0
        \end{tabular}
    \end{center}
This assignation is compatible with the mode algebra of the $bc$- and
$\beta\gamma$-systems. Then a decreasing filtration $F$ on $V$ can be
defined as
\begin{equation}
  \label{eq:Filtration on V}
  F^p V \defeq \{v \in V\ |\  \fdeg v \geq p \},
\end{equation}
which is both exhaustive and weakly convergent:
\begin{equation}
  \bigcup_{p\in\mathbb{Z}} F^p V = V \quad \text{and} \quad \bigcap_{p\in\mathbb{Z}} F^p V = 0.
\end{equation}
The differential $d_{\mathrm{NG}}$ also splits into
$d_{\mathrm{NG}}=d_0+d_1+d_2$ according to this assignation, where $\fdeg
(d_r)=r$, meaning $d_r$ is the part of $d$ that raises $\fdeg$ by
$r$. Explicitly,
\begin{equation}\label{eq:NG differential filtered}
  \dNG= \underbrace{\tfrac{1}{2}\normalord{C\betil}_0}_{d_0} \underbrace{{} -
    \normalord{C \beta\gamma}_0}_{d_1} +  \underbrace{\normalord{C\partial\gamtil}_0}_{d_2} = d_0 + d_1 + d_2.
\end{equation}
The condition $d^2=0$ decomposes into the following:
\begin{equation}
  d_0^2=d_0 d_1 + d_1 d_0 = d_1^2 + d_0 d_2 + d_2 d_0 = d_1 d_2 + d_2 d_1 = d_2^2 =0,
\end{equation}
which tells us that $\dNG$ indeed preserves the filtration (i.e.,
$d(F^p V) \subseteq F^p V$).

We may thus apply Theorem~\ref{thm:Spectral Sequence Converges to
  Cohomology} to deduce that there exists a spectral sequence with
first page
\begin{equation}
    E^{p,q}_1 \cong H^{p+q} \left( \Gr_p V, d_0\right)
\end{equation}
converging to $H^{\dotr}(V,d)$:
\begin{equation}
    E^{p,q}_\infty \cong E^{p,q}_0 \big(H^{\dotr}(V,\dNG),F)\big) = \Gr_p H^{p+q}(V,\dNG).
\end{equation}
Hence, the first thing to do is to compute $d_0$-cohomology to obtain
$E^{p,q}_1$. But this is precisely the setting of
Lemma~\ref{lem:Koszul chiral ring}, with $\lambda=\mu=1$. Since $d_0$
only acts on
$\mathcal{V}\defeq V^{\betil\gamtil}_{\tilde{\sigma}}\otimes V^{BC}$, $V^{\beta\gamma}_{\sigma}$ simply goes along for the
ride. Thus, we have that
\begin{equation} \label{eq:NG cohomology d_0}
    H_{d_0} \cong H^{0,0}_{d_0} \cong E^{0,0}_1 \cong \CC (\ket{\tilde{\sigma}}_{\betil\gamtil} \otimes \ket{-\tilde{\sigma}}_{BC}) \otimes V^{\beta\gamma}_{\sigma}.
\end{equation}
Equation \eqref{eq:NG cohomology d_0} tells us that $E^{p,q}_1 = 0$
for all $p,q\neq 0$. As a result, we must have that $d_1 \equiv 0$ on
$H_{d_0}$, meaning the spectral sequence collapses
(cf. Definition~\ref{def:spectral_sequence_collapse}) at the first
page and calculating $d_2$-cohomology is unnecessary. This proves
Proposition~\ref{prop:NG cohomology}.

\subsection{BRST cohomology of the Gomis--Ooguri string}\label{sec:GO string cohomology app}

Recall that the holomorphic sector of the Gomis--Ooguri string worldsheet CFT is described by a
weight-$(1,0)$ $\beta\gamma$-system with $\gamma$ compactified and 24
free euclidean bosons $\{\partial X^i\}_{i\in\{2,\dots
  25\}}$. The anti-holomorphic sector is identical, so we only focus on the holomorphic one.

We once again filter the space of states
$V\defeq V^{bc}\otimes V^{\beta\gamma} \otimes \mathcal{F}^X(k)$ by
assigning a filtration degree to each mode as follows.
    \begin{center}
        \begin{tabular}{ c|c c c c c } 
            & $\beta_n$ & $\gamma_n$ & $b_n$ & $c_n$ & $L^X_n$ \\ \hline
            $\fdeg$ & -1 & 1 & -1 & 1 & 0
        \end{tabular}
    \end{center}
The corresponding decreasing filtration
\begin{equation}
    F^p V \defeq \{v\in V | \fdeg \geq p\}
\end{equation}
is both exhaustive and weakly convergent. The BRST differential also
splits into $\dVir=d_0+d_1$ as given by equations \eqref{eq:dVir_0} and \eqref{eq:dVir_1}, resulting in a filtered complex.

Applying Lemma~\ref{lem:Koszul chiral ring}, with $\lambda=1$ and
$\mu=2$, we once again obtain that the $d_0$-cohomology is
one-dimensional in the space of states of $(\beta,\gamma,b,c)$ and
spanned by the choice of vacuum $\ket{\mathrm{vac}}_{\sigma}\defeq
\ket{\sigma}_{\beta\gamma}\otimes \ket{1-\sigma}_{bc}$, where
$\sigma\in\ZZ$. Since
the 24 free bosons are unaffected by $d_0$, we have that
\begin{equation}
    H_{d_0}\cong \CC \ket{\mathrm{vac}}_{\sigma} \otimes \mathcal{F}^X(k).
\end{equation}
By the same arguments as those in the previous subsection,
$d_1\equiv 0$ on $H_{d_0}$, resulting in the collapse of the spectral
sequence.  Hence, we explicitly see the computation of the spectrum of
the Gomis--Ooguri string is independent of the $c=24$ CFT that appears
in its matter sector. The choice of 24 free bosons is only
enforced by the limiting procedure employed in the original sigma
model and manifestly guarantees that the spectrum is unitary and galilean
invariant. However, from the worldsheet perspective, one could have
inserted any $c=24$ CFT. Making the choices summarised in
Table~\ref{tab:matter_content} completes the proofs of
Propositions~\ref{prop:holo sector BRST} and \ref{prop:anti-holo
  sector BRST}.

\subsection{Null gauging cohomology with alternate free field realisation} \label{sec:NG cohomology alt app}

Recall from \eqref{eq:Wakimoto embedding original} that we have an
embedding of the Nappi--Witten currents $\{P^\pm, I, J\}$ into the
free fields $(\d\phi^\pm, \beta, \gamma)$. The full space of states
$V=V^{BC}\otimes V^{\beta\gamma}_{\sigma} \otimes \mathcal{F}^\phi(p_+,p_-)$
is spanned by monomials of the form
\begin{equation}
\label{eq:monomials of full space FFR original}
\begin{split}
 \ket{\psi}&= B_{n_1}\dots B_{-n_\mathcal{B}}C_{-m_1}\dots C_{-m_\mathcal{C}}\ket{\rho}_{BC} \\
    &\otimes \beta_{-r_1}\dots \beta_{-r_\mathfrak{B}} \gamma_{-s_1}\dots\gamma_{-s_\mathfrak{C}}\ket{\sigma}_{\beta\gamma} \\
    &\otimes \alpha^{\mu_1}_{-k_1}\dots \alpha^{\mu_A}_{-k_A} \ket{p_+,p_-},   
\end{split}
\end{equation}
where 
\begin{gather*}
k_1 \geq \dots \geq k_{A} \geq 1,\\
n_1 > \dots > n_{\mathcal{B}}\geq 1-\rho,\ \  m_1 > \dots > m_{\mathcal{C}}\geq \rho\\
r_1 \geq \dots \geq r_{\mathfrak{B}} \geq \sigma+1,\ \ s_1 \geq \dots \geq s_{\mathfrak{C}} \geq -\sigma
\end{gather*}
and $\ket{p_+,p_-}$ is the vacuum vector of the Fock module $\mathcal{F}^\phi(p_+,p_-)$ of the lorentzian bosons with momentum $p_\mu=(p_+,p_-)$, which means $\alpha^{\pm}_0\ket{p_+,p_-} = p_\pm \ket{p_+,p_-}$. Note that we allow $\alpha^+_0$ to act non-trivially despite what \eqref{eq:FFR mode embedding} seems to imply; this is because we would like to consider the realisation \eqref{eq:Wakimoto embedding original} independently of how it embeds (as given by \eqref{eq:field embeddings}) into \eqref{eq:Wakimoto embedding}. However, as we will see, $p_+=0$ will be enforced anyway when taking cohomology with respect to the null gauging differential
\begin{equation}
    d_{\mathrm{NG}} = (j_{NG})_0 = \tfrac{1}{2}\normalord{C\d \phi^-}_0 - \normalord{C\beta\gamma}_0 + \normalord{C\d\phi^+}_0
\end{equation}
Just as we have been doing in the previous subsections, we apply a decreasing, weakly exhaustive and convergent filtration on $V$
\begin{equation}
    \label{eq:Filtration on V NG alt}
    F^p V \defeq \{v \in V\ |\  \fdeg v \geq p \},
\end{equation}
given explicitly by the following filtration degrees compatible with
the mode algebra of $\{B,C,\d\phi^\pm, \beta,\gamma\}$:
\begin{center}
    \begin{tabular}{ c|c c c c c c} 
        & $\d\phi^-_n$ & $\d\phi^+_n$ & $B_n$ & $C_n$ & $\beta_n$ & $\gamma_n$ \\ \hline
        $\fdeg$ & -1 & 1 & -1 & 1 & 0 & 0
    \end{tabular}
\end{center}
Once again, the differential $d_{\mathrm{NG}}$ respects the filtration
(i.e., $d(F^p V) \subseteq F^p V$) and splits as
$d_{\mathrm{NG}}=d_0+d_1+d_2$ according to the filtration degree, where
\begin{equation} \label{eq:alt NG differential splitting}
    d_0 = \tfrac{1}{2}\normalord{C\d\phi^-}_0,\ d_1 = -\normalord{C\beta\gamma}_0,\ d_2= \normalord{C\d\phi^+}_0.
\end{equation}
and
\begin{equation}
    d_0^2=d_0 d_1 + d_1 d_0 = d_1 d_2 + d_2 d_1 = d_2^2 =0.
\end{equation}
It again follows from Theorem~\ref{thm:Spectral Sequence Converges to
  Cohomology} that there exists a spectral sequence converging to $H^{\dotr}(V,d)$:
\begin{equation}
  E^{p,q}_\infty \cong E^{p,q}_0 \big(H^{\dotr}(V,d),F)\big) = \Gr_p H^{p+q}(V,d).
\end{equation}
with first page
\begin{equation}
    E^{p,q}_1 \cong H^{p+q} \left(\Gr_p V, d_0\right).
\end{equation}
Notice that the setup is identical to that in Appendix~\ref{sec:GO
  string cohomology app}. As expected, $\d\phi^\pm$ replace $\betil$
and $\d\gamtil$, respectively, and the filtration and splitting of
$d_{NG}$ are also compatible with this dictionary.

We now start with the computation of $d_0$-cohomology. The
differential $d_0$ only acts non-trivially on $V^{BC}\otimes
\mathcal{F}^\phi(p_+,p_-)$, so $V^{\beta\gamma}_{\sigma}$ is unaffected by these
calculations. For convenience, we choose the
$\mathfrak{sl}_2$-invariant $BC$ vacuum
$\ket{\rho}_{BC}=\ket{0}_{BC}$. As mentioned earlier, the computation
of cohomology is greatly simplified by the presence of diagonalisable
endomorphisms which are $\dNG$-exact. In this case, we have (at least)
two such operators.  The first operator is
\begin{equation} \label{eq:alpha- = 0}
    \alpha^-_0= \dNG B_0 + B_0 \dNG.
\end{equation}
This says that the $\dNG$ cohomology can be computed with the
subcomplex with $p_-=0$.  The second operator is
\begin{equation} \label{eq:LdK = 0}
  -\normalord{B\d C}_0 + \normalord{ \d\phi^+ \d\phi^-}_0 \defeq \mathcal{L}_0 = d_0 \normalord{C\d\phi^+}_0+\normalord{C\d\phi^+}_0 d_0
\end{equation}
Note that $\mathcal{L}_0$ is indeed diagonalisable on the subspace
$V^{BC} \otimes \mathcal{F}^\phi(p_+,p_-)$ on
which $d_0$ acts non-trivially.  We may thus restrict to $\ker
\mathcal{L}_0$, which is spanned by monomials \eqref{eq:monomials of
  full space FFR original} such that
\begin{equation}
    \sum_{i=1}^{\mathcal{B}} n_i + \sum_{i=1}^{\mathcal{C}} m_i + \sum_{i=1}^{\mathcal{A}} k_i + \tfrac{1}{2}p^2 = 0.
\end{equation}
The last term is zero already because $p^2 = 2p_+p_-$ and we are
working in the subcomplex with $p_- = 0$, leaving us only with the
sums on the LHS. Furthermore, $n_i \geq 1$ for all $i\in\{1,\dots,
\mathcal{B}\}$, $k_i\geq 1$ for all $i \in \{1,\dots,\mathcal{A}\}$
and $m_i \geq 0$ for all $i\in\{1,\dots \mathcal{C}\}$, which means
that the only two ways the sums on the LHS can equal zero on the RHS
are either if there exist no summands at all or if exactly one of the 
$m_i$ is zero. In other words,
\begin{equation} \label{eq:H_d0 basis}
  \ker{\mathcal{L}_0} \cap \ker \alpha^-_0 = \left( \CC \ket{0}_{BC}
    \oplus \CC C_0\ket{0}_{BC}\right) \otimes \CC \ket{p_+,p_-=0} \otimes V^{\beta\gamma}_\sigma
\end{equation}
and the differential $d_0$ is identically zero in this
subcomplex. This completes the computation of $H_{d_0}$ and thereby
the first page of our spectral sequence:
\begin{equation}
    H_{d_0} \cong E_1 \cong E^{0,0}_1 \oplus E^{0,1}_1 \cong (\CC\ket{0}_{BC} \oplus \CC C_0\ket{0}_{BC})\otimes \CC \ket{p_+,0}\otimes V^{\beta\gamma}_\sigma.
\end{equation}
Before proceeding with the calculation of $H_{d_1}(H_{d_0})$, let us
highlight they key differences between the calculation we just
performed and the $d_0$ cohomology calculation in Appendix~\ref{sec:NG
  cohomology app}. Firstly, the $\dNG$-exact operators that are
available to use are different for each embedding. In the embedding
given by \eqref{eq:Wakimoto embedding}, the operator analogous to
$\alpha^-_0$ would be $\betil_0$. However, $\betil_0$ is not
diagonalisable over $V^{\betil\gamtil}_{\tilde{\sigma}} \otimes V^{BC}$. On the
other hand, although the operator corresponding to $\mathcal{L}_0$ in 
embedding \eqref{eq:Wakimoto embedding} exists
and is diagonalisable on $V^{\betil\gamtil}_{\tilde{\sigma}} \otimes V^{BC}$, there exists a more
restrictive $d_0$-exact operator that we used instead: namely, the
number operator $N_{\mathrm{tot}}$. However, we cannot construct a number
operator in this embedding \eqref{eq:Wakimoto embedding original} the 
way we did in the proof of Lemma~\ref{lem:Koszul chiral ring}; to do
this, we would have needed a CFT generated by $\phi^+$ and $\d\phi^-$
instead of $\d\phi^\pm$.  Summarising, the $d_0$ cohomology under
embedding \eqref{eq:Wakimoto embedding} reduces to a single copy of
the leftover $\beta\gamma$-system, whereas under the embedding
\eqref{eq:Wakimoto embedding original} we are left with two copies of
it instead.

We now perform the computation of $H_{d_1}(H_{d_0})\cong E_2$. In terms of modes,
\begin{equation}
  d_1 = -\normalord{C\beta\gamma}_0 = -\sum_{l\leq 1} C_l \normalord{\beta\gamma}_{-l} - \sum_{l\geq 0}  \normalord{\beta\gamma}_{-l} C_l.
\end{equation}
We can quickly infer from the form of the basis vectors of $H_{d_0}$ that only the term $-C_0\normalord{\beta\gamma}_0$ acts non-trivially on $H_{d_0}$. Simply put, the image of any endomorphism with $C_{n\neq 0}$ lies in the space of $d_0$-coboundaries, which means any such endomorphisms is the zero map in $d_0$-cohomology. To see how it acts explicitly, we consider a generic monomial in $V^{\beta\gamma}_\sigma$ given by
\begin{equation}
    \ket{\Psi}=\beta_{-r_1}\dots \beta_{-r_\mathfrak{B}} \gamma_{-s_1}\dots\gamma_{-s_\mathfrak{C}}\ket{\sigma}_{\beta\gamma}, \quad
    r_1 \geq \dots \geq r_{\mathfrak{B}} \geq \sigma+1,\ \ s_1 \geq \dots \geq s_{\mathfrak{C}} \geq -\sigma.
\end{equation}
Then the action of $d_1$ on the basis of $H_{d_0}$ is given by
\begin{gather}
    d_1 (\ket{0}_{BC} \otimes \ket{p_+,0}\otimes \ket{\Psi}) = -C_0\ket{0}_{BC} \otimes \ket{p_+,0}\otimes \normalord{\beta\gamma}_0\ket{\Psi} \\
    d_1(C_0\ket{0}_{BC} \otimes \ket{p_+,0}\otimes\ket{\Psi}) = 0
\end{gather}
We may decompose $V^{\beta\gamma}_\sigma$ as
\begin{equation*}
  V^{\beta\gamma}_\sigma = \underbrace{\ker \normalord{\beta\gamma}_0}_{\eqdef \left( V^{\beta\gamma}_{\sigma}\right)_0} \oplus \left( V^{\beta\gamma}_{\sigma} \right)_0^\perp,
\end{equation*}
and thus we have that
\begin{gather}
    \ker d_1 = \left(C_0 \ket{0}_{BC} \otimes \ket{p_+,0}\otimes V^{\beta\gamma}_\sigma\right)  \oplus \left(\CC \ket{0}_{BC} \otimes \CC \ket{p_+,0}\otimes \left( V^{\beta\gamma}_{\sigma} \right)_0\right) \\
    \im d_1 = \CC C_0\ket{0}_{BC} \otimes \CC \ket{p_+,0}\otimes \left( V^{\beta\gamma}_{\sigma}\right)_0^\perp.
\end{gather}
This yields
\begin{equation}
    H_{d_1}(H_{d_0}) \cong E_2 \cong \left(\CC\ket{0}_{BC} \oplus \CC C_0\ket{0}_{BC}\right)\otimes \CC \ket{p_+,0} \otimes \left( V^{\beta\gamma}_{\sigma} \right)_0.
\end{equation}
Indeed, $\normalord{\beta\gamma}$ is what one would call the ghost
current whose zero mode counts the number of $\gamma_n$s minus the
number of $\beta_n$s acting on the vacuum $\ket{\sigma}$ in a monomial. There is, however, a
non-trivial action of $\normalord{\beta\gamma}_0$ on a generic vacuum
$\ket{\sigma}$, given by $\normalord{\beta\gamma}_0\ket{\sigma} =
\sigma \ket{\sigma}$. Thus,
\begin{equation}
    \normalord{\beta\gamma}_0 \ket{\Psi} = (\mathfrak{C}-\mathfrak{B}+\sigma) \ket{\Psi},
\end{equation}
and $\left( V^{\beta\gamma}_{\sigma}\right)_0$ is spanned by monomials $\ket{\Psi}$
such that $\mathfrak{C}-\mathfrak{B}+\sigma = 0$.

We perform the final step in this spectral sequence computation; the
computation of
\begin{equation*}
  H_{d_2} (H_{d_1}(H_{d_0}))\cong H_{d_2}(E_2) \cong E_3.
\end{equation*}
Once again, writing $d_2$ in terms of modes
\begin{equation}\label{eq:d_2 modes}
    d_2=\normalord{C\d\phi^+}=\sum_{l\leq1}C_l\alpha^+_{-l}+\sum_{l\geq2}\alpha^+_{-l}C_l
\end{equation}
and considering its action on the basis of
$H_{d_1}$, we see that only the term proportional to $C_0\alpha^+_0$
acts non-trivially. For any $\ket{\Psi}_0\in \left(
  V^{\beta\gamma}_{\sigma}\right)_0$,
\begin{gather}
    d_2 (\ket{0}_{BC} \otimes \ket{p_+,0}\otimes \ket{\Psi}_0) = C_0\ket{0}_{BC} \otimes p_+\ket{p_+,0}\otimes \ket{\Psi}_0 \\
    d_2(C_0\ket{0}_{BC} \otimes \ket{p_+,0}\otimes\ket{\Psi}_0) = 0.
\end{gather}
From this, we see how the condition $p_+=0$ emerges from $d_2$-cohomology. We now have
\begin{gather}
    \ker d_2 = \left(\CC C_0 \ket{0}_{BC} \otimes \CC \ket{p_+=0,0}\otimes \left( V^{\beta\gamma}_{\sigma}\right)_0 \right)  \oplus \left(\CC \ket{0}_{BC} \otimes \CC \ket{p_+=0,0}\otimes \left( V^{\beta\gamma}_{\sigma}\right)_0 \right) \\
    \im d_2 = \CC C_0\ket{0}_{BC} \otimes \CC \ket{p_+\neq 0,0}\otimes \left( V^{\beta\gamma}_{\sigma}\right)_0,
\end{gather}
which lets us complete the proof of Proposition~\ref{prop:NG cohomology original FFR} with
\begin{equation}
    H_{\mathrm{NG}}\cong H_{d_2}(H_{d_1}(H_{d_0})) \cong E_3 \cong \left(\CC\ket{0}_{BC} \oplus \CC C_0\ket{0}_{BC}\right)\otimes \CC \ket{0,0} \otimes \left( V^{\beta\gamma}_{\sigma}\right)_0.
\end{equation}

\section{Null reduced geometries}
\label{sec:null-reds}

In this appendix we describe the null reduced geometries.

\subsection{Null reduction of lorentzian generalised Nappi--Witten group}
\label{sec:null-redt-lor-NW}

We perform the null reduction of a generalised Nappi--Witten Lie group
whose Lie algebra is $\nw_{2\ell+2}$.  The ad-invariant inner product
is given by equation~\eqref{eq:gen-NW-metric} and the brackets are
given by equation~\eqref{eq:gen-NW-brackets}, where the symplectic
form $\omega$ can always be brought via an isometry to the
block-diagonal form
\begin{equation}\label{eq:symplectic form omega}
  \omega = 
  \begin{pmatrix}
    0  & -\mu_1 &    &     &   &\\
    \mu_1 & 0 &    &     &  & \\
    &   & 0  & -\mu_2   &  & \\
    &   & \mu_2 & 0  &   &\\
    &   &    &     & \ddots \\
    &   &    &     &         & 0  & -\mu_\ell \\
    &   &    &     &         & \mu_\ell & 0
  \end{pmatrix},
\end{equation}
where $\mu_1 \geq \mu_2 \geq \dots \geq \mu_\ell > 0$.

We may parametrise the group by $(x^1, \dots x^{2\ell}, u, v)$ where a
generic group element takes the form
\begin{equation}
    g := g(x^1,\dots x^\ell, u, v) = \exp(x^iP_i)\exp(uD)\exp(vZ).
\end{equation}
The pull-back to the parameter space of the left-invariant Maurer--Cartan one-form is
\begin{equation}\label{eq:g-1dg}
   g^{-1}dg = du D + dv Z + \Ad_{\exp(-u D)} \exp(-x^i P_i) d \exp(x^i P_i).
\end{equation}
The derivative of the exponential map is given by (see, e.g., \cite[§1.2,Thm.~5]{MR1889121})
 \begin{equation}
     \exp(-x^i P_i) d \exp(x^i P_i) = \mathscr{D}(\ad_{-x^i P_i})(dx^iP_i),
 \end{equation}
 where 
 \begin{equation}
     \mathscr{D}(z)=\frac{1-\exp(-z)}{z} = 1 - \tfrac{1}{2}z +\mathcal{O}(z^2).
 \end{equation}
 Since $\ad^2_{P_i} = 0$ for all $i\in\{1,\ldots, 2\ell\}$, we are left with
 \begin{equation}\label{eq:derivative of exp(x^i P_i) term}
     \exp(-x^i P_i) d \exp(x^i P_i) = dx^i P_i - \tfrac{1}{2}[-x^i P_i, dx^j P_j] = dx^i P_i + \tfrac{1}{2} x^i dx^j \omega_{ij} Z.
 \end{equation}
Substituting \eqref{eq:derivative of exp(x^i P_i) term} back into \eqref{eq:g-1dg} gives 
\begin{equation}\label{eq:g-1dg penultimate}
    g^{-1}dg = \Ad_{\exp(-uD)}dx^i P_i + du D + \left(dv +  \tfrac{1}{2}x^i dx^j\omega_{ij} \right) Z.
\end{equation}
Next, we use that $\Ad_{\exp(-uD)} = \exp(-u \ad_D)$. Note that because
\begin{equation}
  \ad_D P_i = \omega_{ij}P_j,
\end{equation}
$\ad_D$ is block diagonal.  Let us concentrate on a
single block, say the block spanned by $P_1$ and $P_2$.  Then
\begin{equation}
  \ad_D P_i = \mu_1 \epsilon_{ij} P_j \qquad\text{and hence}\qquad
  \ad_D^2 P_i= - \mu_1^2 P_i.
\end{equation}
Exponentiating,
\begin{equation}
  \exp(-u \ad_D) (dx^1 P_1 + dx^2 P_2) = \left( \cos(u \mu_1) dx^1 +
    \sin(u \mu_1) dx^2\right) P_1+ \left( \cos(u \mu_1) dx^2 - \sin(u \mu_1) dx^1 \right)P_2.
\end{equation}
Putting all the blocks together we arrive at
\begin{multline}\label{eq:MC-one-form}
   g^{-1}dg = du D + \left(dv + \tfrac{1}{2}x^i dx^j
     \omega_{ij}\right)Z \\ + \sum_{i = 1}^\ell \left( \left( \cos(u \mu_i) dx^{2i-1}
       +  \sin(u \mu_i) dx^{2i} \right) P_{2i-1} + \left( \cos(u
       \mu_i) dx^{2i} - \sin(u \mu_i) dx^{2i-1} \right) P_{2i}\right).
\end{multline}
Finally, using \eqref{eq:MC-one-form} and \eqref{eq:gen-NW-metric}, we
obtain the explicit form of the metric on the generalised Nappi--Witten
group
\begin{equation}
  \label{eq:metric-on-NW}
    \eta := \langle g^{-1}dg, g^{-1}dg\rangle = dx^i dx^j \delta_{ij} + 2dudv + x^i \omega_{ij} dx^j du.
\end{equation}

When $\ell=1$, we may set $\mu_1 = 1$ by reciprocal rescalings of the
null coordinates $u$ and $v$ and hence we arrive at the metric of the
target spacetime of the original Nappi--Witten string in terms of the
coordinates $(x^1, x^2, u, v)$:
\begin{equation}\label{eq:metric on H4}
  \eta = 2dudv + (x^1 dx^2 - x^2 dx^1) du + (dx^1)^2 + (dx^2)^2.
\end{equation}

We now consider the reduction of this four-dimensional lorentzian Lie
group by the null central Killing vector $\xi
= \partial_v$.  This generates a one-parameter central subgroup which
simply translates the coordinate $v$ and leaves $u, x^1, x^2$ alone.
Let $M$ denote the quotient manifold, which we can coordinatise by $u,
x^1, x^2$ and let $\pi \colon G \to M$ denote the quotient from the
Nappi--Witten group $G$ to $M$.  The one-form metrically dual to $\xi$
is given by $\xi^\flat = du$, which is both horizontal (since $\xi$ is
null) and invariant.  Therefore $\xi^\flat = \pi^* \tau$ is the
pull-back of a one-form $\tau \in \Omega^1(M)$. Relative to the
coordinates $u,x^1, x^2$ for $M$, $\tau = du$. This is the clock
one-form of the galilean structure on $M$. The spatial cometric
$\gamma \in \Gamma(\odot^2 TM)$ is a symmetric bilinear form on
one-forms.  If $\alpha,\beta \in \Omega^1(M)$, let
$\gamma(\alpha,\beta) \in C^\infty(M)$ be defined by
\begin{equation}
  \pi^* \gamma(\alpha,\beta) = \eta((\pi^*\alpha)^\sharp, (\pi^*\beta)^\sharp ).
\end{equation}
Explicitly, we write
\begin{equation}
  \eta = du \otimes dv + dv \otimes du + \tfrac12 x^1(dx^2 \otimes du
  + du \otimes dx^2) - \tfrac12 x^2(dx^1 \otimes du
  + du \otimes dx^1) + dx^1 \otimes dx^1 + dx^2 \otimes dx^2
\end{equation}
and compute, in addition to $\partial_v^\flat = du$, that
\begin{equation}
  \partial_u^\flat = dv + \tfrac12 (x^1 dx^2 - x^2 dx^1), \quad
  \partial_{x^1}^\flat = dx^1 - \tfrac12 x^2 du \quad\text{and}\quad
  \partial_{x^2}^\flat = dx^2 + \tfrac12 x^1 du.
\end{equation}
This allows us to invert $\flat$ in order to obtain
\begin{equation}
  (du)^\sharp = \partial_v, \quad (dx^1)^\sharp = \partial_{x^1} +
  \tfrac12 x^2 \partial_v \quad\text{and}\quad (dx^2)^\sharp = \partial_{x^2} -
  \tfrac12 x^1 \partial_v
\end{equation}
in addition to a more complicated expression for $(dv)^\sharp$ which
we will not need in what follows.  This allows us to compute the
spatial cometric $\gamma$:
\begin{equation}
  \gamma = \partial_{x^1} \otimes \partial_{x^1} + \partial_{x^2}
  \otimes \partial_{x^2}.
\end{equation}
The three-dimensional galilean manifold $(M,\tau,\gamma)$ is the null
reduction of the Nappi--Witten group.

By construction, this galilean manifold is homogeneous, so we should
be able to identify it, based on the classification in
\cite{Figueroa-OFarrill:2018ilb}.  The original lorentzian metric on
$G$ is bi-invariant, but because of the existence of a nontrivial
centre, the symmetry group $G \times G$ does not act effectively.
Letting $Z < G$ denote the centre, the inherited\footnote{Of course,
  the true symmetry group of the galilean structure is
  infinite-dimensional, but these are ones which come 
  from the isometries of the original lorentzian manifold.} symmetry group of
the null reduction is $G/Z \times G/Z$.  Its Lie algebra consists of
generators $D, D', P_i, P'_i$ where $i=1,2$, with nonzero brackets
\begin{equation}
  [D, P_i] = \epsilon_{ij} P_j \qquad\text{and}\qquad [D', P'_i] =
  \epsilon_{ij} P'_j.
\end{equation}
The generator $D + D'$ acts transitively on the circle of spatial
directions and hence the Lie algebra is a kinematical Lie algebra with
spatial isotropy for a three-dimensional spacetime.  Such kinematical
Lie algebras were classified in \cite{Andrzejewski:2018gmz} and the
Lie algebra above corresponds to the ``euclidean Newton'' algebra in
Table~1 in that paper (or Lie algebra \#24 in Table~8 in
\cite{Figueroa-OFarrill:2018ilb}).  As shown in
\cite[§3.3.5]{Figueroa-OFarrill:2018ilb}, that Lie algebra admits two
spatially isotropic three-dimensional Klein pairs: one is the galilean
limit of anti de~Sitter spacetime and the other is the static
aristotelian spacetime.  To determine which one it is, we need to
identify the stabiliser.  It is not difficult to work out the vector
fields $\xi_X$ on $M$ to which $X \in \{D,D',P_i,P'_i\}$ are sent:
\begin{equation}
  \xi_D = \partial_u, \quad \xi_{D'} = -\partial_u - \epsilon_{ij} x^i
  \partial_{x^j}, \quad \xi_{P'_i} = \partial_{x^i}
  \quad\text{and}\quad \xi_{P_i} = \cos u \partial_{x^i} + \sin u
  \epsilon_{ij} \partial_{x^j}.
\end{equation}
The vector fields $\xi_D + \xi_{D'}$ and $\xi_{P_i}-\xi_{P'_i}$ vanish at the
``origin'', taken to be the point with coordinates $u=x^i = 0$.  These
corresponds to the Lie algebra of the stabiliser of that point.  We
see by inspection that there is no ideal of the transitive Lie algebra
which is contained in the stabiliser, hence the Klein pair is
effective and hence it coincides with the Klein pair of
three-dimensional galilean AdS  (which is \#39 in Table~9 in
\cite{Figueroa-OFarrill:2018ilb}).

\subsection{Null reduction of split generalised Nappi--Witten group}
\label{sec:null-red-split-NW}

In the bulk of the paper we used a different real form of the
complexified Nappi--Witten Lie algebra in which the ad-invariant inner
product has split signature.  A basis is now given by $P_i^+, P_i^-,
D, Z$, where $i=1,\dots,\ell$.  The nonzero Lie brackets are now
\begin{equation}
  [D, P_i^\pm] = \pm \mu_i P_i^\pm \qquad\text{and}\qquad [P_i^+, P_j^-] = \mu_i \delta_{ij} Z,
\end{equation}
where there is no summation in the second expression.  The
ad-invariant inner product is given by
\begin{equation}
  \left<P_i^+, P_j^-\right> = \delta_{ij} \qquad\text{and}\qquad \left<D,Z\right> = 1.
\end{equation}
We parametrise the group with null coordinates $(x^i, y^i, u, v)$,
with $i =1,\dots,\ell$, where the generic group element is
\begin{equation}
  g := g(x^i,y^i,u,v) = \exp(x^i P_i^+ + y^i P_i^-) \exp(u D) \exp(v Z).
\end{equation}
The pull-back of the left-invariant Maurer--Cartan one-form is given
by
\begin{equation}
  g^{-1}dg = dv Z + du D + \Ad_{\exp(-uD)} \exp(-x^i P^+_i - y^i P^-_i) d \exp(x^i P^+_i + y^i P^-_i).
\end{equation}
As before, we may calculate
\begin{equation}
  \exp(-x^i P^+_i - y^i P^-_i) d \exp(x^i P^+_i + y^i P^-_i) = dx^i
  P^+_i + dy^i P^-_i + \tfrac12 \delta_{ij} (x^i dy^j - y^i dx^j) \mu_i Z
\end{equation}
and
\begin{equation}
  \Ad_{\exp(-u D)} \left( dx^i  P^+_i + dy^i P^-_i \right) = e^{-u \mu_i} dx^i P^+_i +  e^{u \mu_i} dy^i P^-_i,
\end{equation}
so that
\begin{equation}
  g^{-1}dg = \left( dv + \tfrac12 (x^i dy^j - y^i dx^j)\delta_{ij}  \right) Z + du D + e^{-u \mu_i} dx^i P^+_i +  e^{u \mu_i} dy^i P^-_i.
\end{equation}
The bi-invariant metric on the split generalised Nappi--Witten group
is then
\begin{equation}
  \eta = \left<g^{-1}dg, g^{-1}dg\right> = 2 du dv +  \sum_{i=1}^\ell \mu_i (x^i dy^i - y^i dx^i) du + 2 \sum_{i=1}^\ell dx^i dy^i.
\end{equation}

For $\ell = 1$, we may set $\mu_1 = 1$ by a reciprocal rescaling of
$u$ and $v$ and arrive at
\begin{equation}
  \eta = 2 du dv +  (x dy - y dx) du + 2 dx dy
\end{equation}
relative to null coordinates $(u,v,x,y)$.

The four-dimensional split Nappi--Witten Lie algebra can be understood
either as a double extension of the trivial Lie algebra by the nonabelian
two-dimensional Lie algebra or, equivalently, as an iterated double
extension of the trivial Lie algebra by one-dimensional Lie algebras.

In the latter description, we consider an abelian two-dimensional
lorentzian Lie algebra.  Relative to a Witt (a.k.a. ligthcone) frame
we may view the Lie
algebra as $\h \oplus \h^*$, where $\h$ is a one-dimensional Lie
algebra and the dual pairing between $\h$ and $\h^*$ provides the
(trivially ad-invariant) lorentzian inner product.  Let $\h = \RR P_+$
and $\h^* = \RR P_-$ with $\left<P_+,P_-\right> = 1$.  Let $D$ be the
skew-symmetric derivation defined by $D P_\pm = \pm P_\pm$.  We then
double extend $\h \oplus \h^*$ by the one-dimensional Lie algebra
$\RR D$ acting via $[D,P_\pm] = \pm P_\pm$ and letting $Z$ span the
dual of $\RR D$, we have that $[P_+,P_-]=Z$ and $\left<D,Z\right>=1$.

In the former description, we let $\a$ be the two-dimensional
non-abelian Lie algebra with basis $X,Y$ and bracket $[X,Y]=Y$.  Let
$\a^*$ denote the dual with canonical dual basis $\alpha,\beta$.  The
double extension of the trivial Lie algebra by $\a$ is the semidirect
product $\a \ltimes \a^*$ with $\a^*$ an abelian ideal.  The action of
$\a$ on $\a^*$ is via the coadjoint representation:
\begin{equation*}
  [X,\alpha] = \ad^*_X \alpha = 0, \quad   [Y,\alpha] = \ad^*_Y \alpha
  = 0, \quad [X,\beta] = \ad^*_X \beta = - \beta \quad\text{and}\quad
  [Y,\beta] = \ad_Y^* \beta = \alpha.
\end{equation*}
Then letting $X = D$, $Y= P_+$, $\alpha = Z$ and $\beta = P_-$
provides the desired isomorphism with the split Nappi--Witten Lie
algebra.

\end{appendices}

\newpage
\bibliographystyle{utphys}
\bibliography{References}
 
\end{document}